\documentclass[11pt,reqno]{amsart}
\usepackage[letterpaper,margin=1in]{geometry}               
\usepackage{amssymb,amsmath,amsfonts,stmaryrd}
\usepackage{xcolor}
\usepackage{quiver}
\usepackage[utf8]{inputenc}

\usepackage{graphicx}
\usepackage{multicol}
\usepackage{tikz-cd}
\usepackage{lipsum}
\usepackage{url}
\usepackage{adjustbox}
\usepackage{tikz}
\usepackage{bbold}
\usepackage{placeins}
\usepackage{float}
\usepackage{tikz-cd}
\usetikzlibrary{shapes.geometric,decorations.markings,arrows,decorations.pathreplacing,calc,3d}
\usepackage{arydshln}
\newcommand{\boundellipse}[3]% center, xdim, ydim
{(#1) ellipse (#2 and #3)
}

\usepackage{tikz}
\usetikzlibrary{cd}
\usetikzlibrary{calc}
\usetikzlibrary{decorations,decorations.pathreplacing,decorations.markings,decorations.pathmorphing}
\tikzset{squiggly/.style={decorate, decoration=snake}}

\tikzset{super thick/.style={line width=3pt}}
\tikzstyle{far>}=[decoration={markings, mark=at position 0.75 with {\arrow{>}}}, postaction={decorate}]
\tikzstyle{mid>}=[decoration={markings, mark=at position 0.55 with {\arrow{>}}}, postaction={decorate}]
\tikzstyle{mid<}=[decoration={markings, mark=at position 0.55 with {\arrow{<}}}, postaction={decorate}]
\tikzset{super thick/.style={line width=3pt}}
\tikzstyle{far>}=[decoration={markings, mark=at position 0.75 with {\arrow{>}}}, postaction={decorate}]
\tikzstyle{mid>}=[decoration={markings, mark=at position 0.55 with {\arrow{>}}}, postaction={decorate}]
\tikzstyle{mid<}=[decoration={markings, mark=at position 0.55 with {\arrow{<}}}, postaction={decorate}]
\tikzstyle{knot}=[preaction={super thick, white, draw}]
\tikzstyle{coupon}=[draw, very thick, rectangle, rounded corners=5pt]
\tikzset{Rightarrow/.style={double equal sign distance,>={Implies},->},
triplecd/.style={-,preaction={draw,Rightarrow}},
quadruplecd/.style={preaction={draw,Rightarrow,
shorten >=0pt
},
shorten >=1pt,
-,double,double
distance=0.2pt}}
\tikzset{
    tripleline/.style args={[#1] in [#2] in [#3]}{
        #1,preaction={preaction={draw,#3},draw,#2}
    }
}
\tikzstyle{triple}=[tripleline={[line width=.15mm,black] in
      [line width=.7mm,white] in
      [line width=1mm,black]}] 
\tikzset{
    quadrupleline/.style args={[#1] in [#2] in [#3] in [#4]}{
        #1,preaction={preaction={preaction={draw,#4},draw,#3}, draw,#2}
    }
}
\tikzstyle{quadruple}=[quadrupleline={[line width=.3mm,white] in
      [line width=.6mm,black] in
      [line width=1.2mm,white] in
      [line width=1.5mm,black]}]
\usepackage{comment}

\usepackage{arydshln}

\newcommand\arXiv[1]{\href{http://arxiv.org/abs/#1}{\nolinkurl{arXiv:#1}}}
\newcommand\MRnumber[1]{\href{http://www.ams.org/mathscinet-getitem?mr=#1}{\nolinkurl{MR#1}}}
\newcommand\DOI[1]{\href{http://dx.doi.org/#1}{\nolinkurl{DOI:#1}}}
\newcommand\MAILTO[1]{\href{mailto:#1}{\nolinkurl{#1}}}

\usepackage[colorlinks,
             linkcolor=blue,
             citecolor=green!70!black,
             pdfproducer={pdfLaTeX},
             pdfpagemode=UseNone,
             bookmarksopen=true,
             bookmarksnumbered=true]{hyperref}
\usepackage{cleveref}

\newcounter{mainthm}

\newtheorem{maintheorem}[mainthm]{Theorem}
\newtheorem{dummy}{Dummy}[section]

\newtheorem{lemma}[dummy]{Lemma}
\newtheorem{proposition}[dummy]{Proposition}
\newtheorem{corollary}[dummy]{Corollary}

\newtheorem{theorem}[dummy]{Theorem}
\theoremstyle{definition}
\newtheorem{definition}[dummy]{Definition}
\newtheorem{warning}[dummy]{Warning}
\newtheorem{construction}[dummy]{Construction}
\newtheorem{notation}[dummy]{Notation}
\newtheorem{claim}[dummy]{Claim}
\newtheorem{assumption}[dummy]{Premise}

\theoremstyle{definition}

\newtheorem{rem}[dummy]{Remark}

\newtheorem{example}[dummy]{Example}

\usepackage{dsfont}
\renewcommand\mathbb\mathds

\newcommand\Z{\mathbb Z}

\newcommand\cA{\mathcal A}

\newcommand\cC{\mathcal C}
\newcommand\cD{\mathcal D}

\newcommand\cF{\mathcal F}

\newcommand\cL{\mathcal L}
\newcommand\cM{\mathcal M}
\newcommand\cN{\mathcal N}

\newcommand\cT{\mathcal T}

\newcommand\cZ{\mathcal Z}

\newcommand\rB{\mathrm B}

\newcommand\rH{\mathrm H}

\newcommand\rL{\mathrm L}

\newcommand\rU{\mathrm U}

\newcommand\rZ{\mathrm Z}

\newcommand\Sp{\mathbf{Sp}}
\newcommand{\cn}{\mathrm{cn}}
\newcommand{\st}{\mathrm{st}}

\newcommand\pt{\mathrm{pt}}
\usepackage[mathscr]{euscript}

\newcommand{\FZ}{\text{\usefont{U}{euf}{m}{n}Z}}

\DeclareMathOperator\homology{H}
\renewcommand\H{\homology}

\newcommand{\dev}[1]{{\bf \color{blue} [DS: #1]}}

\newcommand{\GL}{\mathrm{GL}}

\newcommand{\id}{\mathrm{id}}

\newcommand{\Mor}{\mathbf{Mor}}

\newcommand{\bTheta}{\mathbf{\Theta}}

\newcommand{\op}{\mathrm{op}}

\newcommand{\Bord}{\cat{Bord}}

\newcommand{\Aff}{\mathbf{Aff}}

\newcommand\longto\longrightarrow
\newcommand\mono\hookrightarrow
\newcommand\epi\twoheadrightarrow

\newcommand\<\langle
\renewcommand\>\rangle
\newcommand\sminus\smallsetminus

\DeclareFontFamily{U}{min}{}
\DeclareFontShape{U}{min}{m}{n}{<-> udmj30}{}
\newcommand\yo{\!\text{\usefont{U}{min}{m}{n}\symbol{'207}}\!}

\DeclareMathOperator\End{End}

\DeclareMathOperator\Spec{Spec}

\newcommand\Vect{\mathbf{Vect}}
\newcommand\tVect{\mathbf{2Vect}}
\newcommand\thVect{\mathbf{3Vect}}
\newcommand\nVect{\mathbf{nVect}}

\newcommand{\npVect}{\mathbf{(n+1)Vect}}
\newcommand\nQcoh{\mathbf{nQCoh}}
\newcommand\noQcoh{\mathbf{(n-1)QCoh}}
\newcommand{\npQcoh}{\mathbf{(n+1)QCoh}}
\newcommand\Rep{\cat{Rep}}
\newcommand\nRep{\mathbf{nRep}}
\newcommand{\npRep}{\mathbf{(n+1)Rep}}
\newcommand\Mod{\cat{Mod}}

\newcommand{\CAlg}{\mathbf{CAlg}}
\newcommand{\Alg}{\mathbf{Alg}}

\newcommand{\Bimod}{\mathbf{Bimod}}
\newcommand{\Sch}{\mathbf{Sch}}
\newcommand{\Stk}{\mathbf{Stk}}
\newcommand{\prstk}{\mathbf{PStk}}

\newcommand{\fib}{\mathrm{Fib}}
\newcommand{\cofib}{\mathrm{coFib}}
\newcommand{\Qcoh}{\mathbf{QCoh}}
\newcommand{\Corr}{\mathbf{2Corr}}
\newcommand{\nCorr}{\mathbf{nCorr}}

\newcommand{\Cat}{\mathbf{Cat}}

\renewcommand{\Pr}{\mathbf{Pr}}
\newcommand{\nPr}{\mathbf{nPr}}

\newcommand{\Ab}{\mathbf{Ab}}
\newcommand{\Fun}{\mathrm{Fun}}
\newcommand{\Set}{\mathbf{Set}}
\newcommand{\Fin}{\mathbf{Fin}}
\newcommand{\Spaces}{\infty\mathbf{Gpd}}

\DeclareMathOperator{\colim}{colim}
\let\lim\relax
\DeclareMathOperator{\lim}{lim}

\newcommand{\Sh}{\mathbf{Sh}}

\newcommand{\Hom}{\mathrm{Hom}}

\newcommand{\fd}{\mathrm{fd}}
\newcommand\cat[1]{\mathbf{#1}}

\title{Geometric Categories for Continuous Gauging}
\date{}

\begin{document}

\author{Devon Stockall}
\address{\textnormal{\url{stockall@imada.sdu.dk}}, University of Southern Denmark}

\author{Matthew Yu}
\address{\textnormal{\url{yumatthew70@gmail.com}}, University of Oxford}

%\author{Devon Stockall and Matthew Yu}
\begin{abstract}
%We present a unified categorical framework which encodes gauging of continuous and finite symmetries in arbitrary spacetime dimension.  We show how our framework can identify electric and magnetic symmetries expected in $G$-gauge theory, and how to capture electric symmetry breaking resulting from the addition of charged matter. We work with geometric categories, i.e. categories internal to stacks.  This allows us to extend (de)equivariantization of fusion categories to continuous groups, construct a functorial SymTFT and boundaries for this theory, and compute the relevant categories of endomorphisms and Drinfeld centers.  

We develop a unified categorical framework for gauging both continuous and finite symmetries in arbitrary spacetime dimensions. Our construction applies to geometric categories i.e. categories internal to stacks.  This generalizes the familiar setting of fusion categories, which describe finite group symmetries, to the case of Lie group symmetries. Within this framework, we obtain a functorial Symmetry Topological Field Theory together with its natural boundaries, allowing us to compute associated endomorphism categories and Drinfeld centers in a uniform way. For a given symmetry group $G$, our framework recovers the electric and magnetic higher-form symmetries expected in $G$-gauge theory. Moreover, it naturally encodes electric breaking symmetry in the presence of charged matter, reproducing known physical phenomena in a categorical setting. 
%In this way, our approach extends (de)equivariantization beyond the finite setting, providing a single geometric and functorial language for continuous and discrete gauge symmetries, their duals, and their categorical boundary data.
\end{abstract}
\maketitle
%\phantom{aaaaaaaaaaa}
\tableofcontents

%%%%%%%%%%%%%%%%%%%%%%%%%%%%%%%%%%
\section{Introduction}\label{section:intro}
%%%%%%%%%%%%%%%%%%%%%%%%%%%%%%%%%%

Finite symmetries of a QFT are well captured by higher fusion categories.  Many symmetries of interest in physics are not finite: Lie group symmetry is essential for gauge theory.  The inherently finite nature of fusion categories makes them ill-suited to capture these situations, where one expects a manifold of objects.  The apparent generalizations also fail: linear monoidal categories fail to capture the higher homotopy of Lie groups.  $\infty$-categories do not capture finer geometry: one can not distinguish $\infty$-categories with objects labeled by $\rU(1)$ or by $\rB \Z$.  To view continuous symmetries in a categorical setting, we must consider \emph{geometric} categories: categories whose collection of objects are equipped with geometric structure, as suggested in \cite{Schreiber:2013pra,Gwilliam:2025vdu}.  Working with geometric categories has two key advantages: objects and morphisms can be organized into families, rather than treated individually, and standard categorical manipulations manifestly respect the geometry, making them well suited to treat continuous symmetry.

Similar categories have been used to great effect in \cite{Freed:2009qp,Weis:2022egw,Jia:2025vrj}, where continuous symmetries are modeled by the category of skyscraper sheaves over a topological group $G$, and section 8 of \cite{DimofteSparks}, where symmetries in 3d $\cN=4$ gauge theory are modeled by the category of quasi-coherent sheaves.  Continuous Tambara-Yamagami categories were studied in \cite{Marin:2025stc}.  One would hope that continuous categories can help to recover dual symmetries which arise from dynamically gauging a continuous symmetry, similar to the fusion-categorical approach for finite groups \cite{Tachikawa:2017gyf}.  We initiate a series of works aimed at developing the categorical framework necessary to give a unified treatment of continuous and finite symmetries and their gauging.  

\begin{enumerate}
    \item We propose a category which can accommodate geometric aspects of continuous $G$ symmetries, including operators continuously supported on $G$. 
    \item We construct a SymTFT from this category, and study its topological boundaries to extract symmetries of $G$ gauge theories.
\end{enumerate}

\iffalse
Dynamical gauging by continuous group is not a completely topological manipulation, with the effect of modifying the original Hilbert space by introducing new propagating degrees of freedom. This is unlike in the case of gauging finite groups, where all connections for the group are flat, and thus is a topological manipulation.
One of the main difficulties one faces when studying continuous symmetries is the plethora of choices for categories that one might use to describe them.  
\fi

Our model for categorical continuous symmetry will be $\nQcoh(X)$, the category of (higher) quasi-coherent sheaves on an algebraic stack $X$.  To justify this choice, we consider the classifying stack $\rB G$ of an algebraic group $G$, implementing a 0-form $G$-global symmetry.  This formalism has the following physical advantages: 
\begin{enumerate}
    \item An equivariantization procedure for continuous group $G$ is implemented in arbitrary dimension using the notion of \textit{higher affineness} \cite{Gaitsgory_2015,stefanich:thesis}. 
    \item Electric symmetry breaking is computed using the \emph{integral transform} formulas of \cite{ben2010integral,BFNMorita,stefanich2023,StefanichSheaf}. 
    \item Magnetic symmetry can be associated to invertible module categories of $\Qcoh(\rB G)$, and so are classified using \cite{StefanichLinear}. 
    %, which leads to magnetic symmetries on the topological boundary of the SymTFT.
\end{enumerate}
Gauging in this setting recovers the dual electric and magnetic symmetries expected in flat $G$ gauge theory. 
  
%%%%%%%%%%%%%%%%%%%%%%%%%%%%%%%%%%
\subsection{The External Perspective on Symmetries}\label{subsection:external}
%%%%%%%%%%%%%%%%%%%%%%%%%%%%%%%%%%
Categories for continuous symmetries are unwieldy to study directly via their objects, morphisms, and fusion rules.  We instead present an ``external'' perspective on symmetries, whereby we study theories as objects in a ``category of QFTs'', in the spirit of  \cite{Kong:2015flk,Gaiotto:2015zna,Kong:2020iek,Kong:2021ups,GJF2019,Zeev:2022cnv,costello2023factorization,StockallCondensation,Stockall2025}.  We will denote the category of QFTs by $\bTheta$.  Its objects are theories, 1-morphisms are topological domain walls between theories, and composition given by fusion of domain walls.  Higher morphisms can also be considered, with $i$-morphisms given by codimension-$i$ topological defects supported on the interface between codimension-$(i-1)$ defects.  This makes the category of $(n+1)$d-theories $\bTheta$ into a \emph{geometric $(\infty,n+1)$-category}.  Given a theory $\nu \in \bTheta$, one can completely recover topological defects from that theory to itself by the endomorphism $n$-category $\Omega_\nu (\bTheta):=\End_\bTheta(\nu)$.  The category $\Omega_\nu (\bTheta)$ is the \emph{category of generalized global symmetries} of $\nu$ in the sense of \cite{GeneralizedGlobalSymm}.  To build familiarity with this perspective, we will compare with symmetries recovered from the SymTFT and the cobordism hypothesis.  

%That is, we consider a category of $(n+1)$-dimensional theories given by $\Mod_{\nVect}$. Our gauging prescription proceeds by constructing $\nRep( G)$ and analyzing symmetries that appear in gauge theories symmetries through its category of modules $\Mod_{\nRep(G)}$. To determine whether a given symmetry can be realized in this framework, we construct a functor from the classifying space of that symmetry into $\Mod_{\Rep( G)}$. Such a functor selects a module category that exhibits the corresponding symmetry. 

\subsubsection{From cobordism to Morita categories}
An $(n+1)$-dimensional framed linear TQFT is given by a symmetric monoidal functor $$\cZ:\Bord^{\mathrm{fr}}_{(n+1)}\to \npVect,$$ where $\npVect=\Mod_{\nVect}$ is defined inductively.  A codimension $1$-defect is given by a \emph{lax natural transformation} between TQFT functors.  Additionally, a codimension $i$-defect is given by a lax $i$-transform.  By cobordism hypothesis with singularities \cite[section 4.3]{LurieTFT} and \cite[Corollary 7.7]{JFS2017}, there is then an equivalence of categories between the category of TQFTs described in this way, and a subcategory $\npVect^{\mathrm{fd}}\subset \npVect$ on those objects and $i$-morphisms which are sufficiently dualizable.  Then operators in a TQFT are captured equally well by simply studying the target category $\npVect$.  There is a string of fully faithful inclusions\footnote{For the right inclusion, see \cite[Section 4.1]{StockallCondensation}.  The left inclusion will appear in upcoming work of Markus Zetto and David Reutter.}
\begin{equation}
    \npVect^{\fd} \subset \Mor(\nVect) \subset  \npVect,
\end{equation}
 where $\Mor(\nVect)$ is the \emph{Morita category} of $\nVect$, whose objects are algebras in $\nVect$, morphisms are bimodules, 2-morphisms are bimodule morphisms, etc.  Using these inclusons, fully dualizable objects of $\thVect$ are identified with \emph{fusion categories} \cite{DSPS}.  Higher dimensional analogs of this construction are explored in \cite{BJS2018,DouglasReutter2018,GJF2019,JohnsonFreydTO,Johnson-Freyd:2021tbq, Decoppet:2024htz,StockallCondensation}.

\iffalse 
\begin{rem}\dev{should I remove this remark?  We really only need DSPS for fusion categories}
 Fully dualizable objects of $\Vect$ are finite dimensional vector spaces.  Using the Morita-category characterization, full dualizable objects of $\tVect$ correspond to \emph{seperable algebras} in $\Mor(\Vect)$.  Fully dualizable objects of $\thVect$ are identified with \emph{fusion categories} in $\Mor(\tVect)$ \cite{DSPS}, see \cite{Decoppet:2024htz,Bhardwaj:2024xcx} for more context on fusion 2-categories and fusion 3-categories.  Classification of such objects in dimension 4 is explored in \cite{BJS2018,DouglasReutter2018}, and constructions of such object in arbitrary dimensions is given by \cite{GJF2019}.  One can define the category of fusion $(n-2)$-categories to be dualizable objects  in the subcategory $\nVect^{\fd}\subset \Mor((n-1)\Vect)$  \cite{JohnsonFreydTO}.  See \cite[Section 4.1]{StockallCondensation} for a coherent treatment of this definition in terms of Cauchy completion.  Definition in terms of idempotent completion will appear in upcoming work of Reutter and Zetto.
\end{rem} 
\fi
 
\subsubsection{From Morita categories to SymTFT}
For the sake of comparison with SymTFT, we will focus on dimension $3$.  Consider a fusion category $\cC\in \Mor(\tVect)$, and denote by $\nu$ the $3$d TQFT constructed out of this category via the cobordism hypothesis.  Using cobordism hypothesis with singularities, we have the following identifications.  A \emph{boundary} for $\nu$ is determined by an object 
\begin{equation}
    \text{bdry}(\nu)\simeq \Hom_{\Mor(\tVect)}(\cC,\Vect)\simeq \Mod_\cC
\end{equation}
The collection of surface operators in $\nu$ is given by 
\begin{equation}
\text{surf}(\nu)\simeq \End_{\Mor(\tVect)}(\cC)\simeq \Bimod_{\cC|\cC}. 
\end{equation}
The collection of line operators in $\nu$ is given by 
\begin{equation}\label{eq:codim2}
\text{line}(\nu)\simeq \End_{\Bimod_{\cC|\cC}}(\cC)\simeq \FZ(\cC), 
\end{equation}
where the final equivalence with the Drinfeld center follows from \cite[Theorem 5.3.1.30]{LurieHA}.

This gives the translation between the cobordism hypothesis and SymTFT/SymTO perspective on boundaries, consistent with \cite[Section 4.3]{FreedMooreTeleman}\footnote{By TO, we mean ``topological order'', which is defined in \cite[Definition 2.21]{Kong:2015flk} and \cite[Definition I.1]{JohnsonFreydTO}. See also where this perspective is detailed  \cite{kong2020algebraic,Chatterjee:2022kxb,Chatterjee:2022jll}.}.  In this way, both settings capture the same symmetries arising from Dirichlet boundary conditions of codimension-2 operators in the bulk.  
Consider a $1+1$d boundary $M\in \text{bdry}(\nu)\simeq \Mod_\cC$ for the $2+1$d theory $\nu$.  Line operators in the boundary $M$ are \emph{endomorphisms} of this boundary, and so correspond objects of $\End_{\Mod_\cC}(M)$.  Then $M$ has fusion categorical symmetry given by the category $\End_{\Mod_\cC}(M)$.  A (condensable) algebra object $A\in \Alg(\End_{\Mod_\cC})$\footnote{Such an algebra object is a \emph{monad} in $\Mod_\cC$.  This motivates the condensation construction given in \cite{GJF2019}. } determines a family of condensable boundary lines.  The fusion categorical symmetry after condensing is given by 
\begin{equation}
    \Bimod_{A|A}(\cC)\simeq \End_{\Mod_\cC}(\Mod_A(\cC)), 
\end{equation}
where the final equivalence follows from \cite[Remark 4.8.4.9]{LurieHA}.  This gives an alternative perspective on condensations, where one
\begin{enumerate}
    \item Begins with a condensable algebra object $A$. 
    \item Constructs a new theory $\Mod_A(\cC)\in \Mod_\cC$ from it. 
    \item Computes the endomorphisms $\End_{\Mod_\cC}(\Mod_A(\cC))$ of the new theory. 
\end{enumerate}
This perspective was suggested in arbitrary dimension by \cite{Kong:2024ykr}, and extended in \cite{StockallCondensation}.  It is often practically easier to calculate in terms of module categories, and then translate back to statements about the Drinfeld center and SymTFT.  For this reason, we will construct an $(n+2)$d SymTFT from a category $\cC$ via the cobordism hypothesis as above.  We will then identify module categories as boundary conditions, and compute their endomorphisms to determine the boundary symmetries. 

%More generally, the collection of codimension $i$-operators in $\nu$ are calculated to be  
%\begin{equation}
%\text{codim-}i(\nu)\simeq  \End^i_{\cC}(\Mor(\noVect^{\fd})), 
%\end{equation}
%where $\End^i_\cC$ denotes taking an $i$-fold endomorphism at $\cC$.

%Having done the last example which We can now connect this external perspective of recovering operators with what the literature does with the SymTFT...

We conclude this section with remarks on higher form symmetries.  One can view an $i$-form $G$ symmetry in a QFT as a family of codimension $(i+1)$ defects parametrized by $G$, with appropriate fusion.  This data is encoded by functors from a classifying stack $B^{i+1}G$.  
\begin{definition}\label{def:iform}
Let $\bTheta$ be a (geometric) $(\infty, n)$-category of $n$-dimensional QFTs.  Let $G$ be a group. A theory with $i$-form $G$ symmetry is a functor $\rB^{i+1}G\to \bTheta$.\footnote{Note that $G$ must be sufficiently monoidal to be delooped $(i+1)$-times. A similar version of this definition appears in \cite{Lan:2023uuq} and \cite[Remark 5.3.15]{Kong:2024ykr}.}
\end{definition}

\begin{example}\label{ex:zeroform}
    A $0$-form $G$-symmetry is a functor $F:\rB G\to \bTheta$.  This picks a theory $F(\pt)\in \bTheta$, and, for each $g\in G$, a codimension-1 defect $F(g)$ from this theory itself.  Functoriality imposes that fusion of the defects agrees with group multiplication: $F(g)F(h)\simeq F(gh)$. 
\end{example}
\begin{example}
     Let $G$ be an abelian group.  A $1$-form $G$-symmetry is a functor $F:\rB^2 G\to \bTheta$ that picks out a theory $F(\pt)$, the invisible defect $\id_{F(\pt)}\simeq F(\id_\pt)$ from $F(\pt)$ to itself, and a genuine codimension-2 defect $F(g)$ for each $g\in G$. Each genuine codimension-2 is viewed as a morphism from the invisible codimension-1 defect to itself. 
\end{example}
\begin{example}
Let $X$ be a geometric space.  A $(-1)$-form $X$-symmetry is a functor $X\to \bTheta$, or equivalently, a family of theories parameterized by $X$.
\end{example}
These examples offer a concrete criterion for identifying $i$-form symmetries in higher-dimensional gauge theories, which we will exploit in later sections.

\subsection{The Gauging Prescription}\label{subsection:mainresults}
%%%%%%%%%%%%%%%%%%%%%%%%%%%%%%%%%%
 We now introduce general results on the gauging of symmetries, applicable in arbitrary spacetime dimensions and suitable for constructing the symmetry structures of gauge theories. Our approach proceeds from an external perspective, which emphasizes that gauging can be formulated as the manipulation of symmetries within an ambient category of physical theories.  This will use the following notation.  Let $\mathds{K}$ be a characteristic 0 field. 
\begin{itemize}
    \item Let $\Vect$ denote the $(\infty,1)$-category of (not necessarily finite) chain complexes of $\mathds{K}$-vector spaces.
    \item Let $\nVect$ denote the $(\infty,n)$-category of linear $(\infty,n-1)$-categories.
    \item Let $G$ be an affine algebraic group over $\mathds{K}$, and $\rB G = \pt/G$ the classifying stack. 
    \item The analog of $\nVect(G)$ is $\nQcoh^*(G)$, the category of higher quasi-coherent sheaves with the convolution monoidal product\footnote{One can also consider versions of $\nQcoh^*(G)$ with monoidal structure twisted by $\omega\in \H^3(\rB G, \mathbb{G}_m)$.  We will not do this here.}.
    \item Let $\nRep(G):=\Fun(\rB G,\underline{\nVect})$ denote the category of $n$-vector spaces with $G$-action.
\end{itemize}

\begin{claim}\label{claim:mainclaim}
Symmetries and symmetry breaking of $G$-gauge theories (possibly with charged matter) are encoded, in $(n+1)$-dimensions, by choices of $\nRep(G)$-module categories.
\end{claim}

\noindent The intuition behind this claim is to view \(\nRep(G)\) as defining a symmetry topological field theory (SymTFT) \cite{Kong:2015flk,Kong:2017hcw,Apruzzi:2021nmk,FreedMooreTeleman}, where a choice of module category for \(\nRep(G)\) specifies a topological boundary condition for the SymTFT. %\matt{there is another subtley which is that people usually mean symTO when they say symTFT, confusingly. we should put a footnote what we mean and that this isnt a topological order} 
The symmetries on the boundary can be identified with endomorphisms of the chosen module category.  In this work, we will use this perspective to reproduce the global symmetry data of flat $G$-gauge theories.  

Upon adding in the data of the ``local quantum symmetries'', which encode dynamical information on the physical boundary (see \cite{Kong:2019byq,Kong:2019cuu,Kong:2020iek} for the terminology), compactifying the SymTFT will produce a particular physical theory.  We expect it is possible to choose local quantum symmetries on the physical boundary to recreate electric and magnetic symmetry of dynamical gauge theories.

The following theorem constructs a partially-defined $(n+2)$-dimensional symmetry TQFT from $\nRep(G)$.
\begin{maintheorem}[{\Cref{thm:generalTQFT}}]\label{mainthm:tqft}
 Let $G$ be an affine algebraic group over characteristic $0$ field $\mathbb{K}$.  Then there is a fully extended once-categorified $(n+1)$-dimensional TQFT associated to $\nRep(G)$, for all $n\geq 1$.\footnote{The reader is directed to read \Cref{warning:shiftindex} before comparing \Cref{thm:generalTQFT} with the contents of this theorem, as there is a shift of index.}
\end{maintheorem}

To support \Cref{claim:mainclaim}, we wish to extract an $\nRep(G)$-module category from every $(n+1)$-dimensional TQFT with $G$-symmetry.  Using cobordism hypothesis, we can construct $(n+1)$-dimensional linear TQFTs from fully dualizable objects of $\npVect$, and so, for our purpose, we take $\npVect$ to be ``a category of $(n+1)$-dimensional QFTs''.  Recall from \Cref{def:iform} that an $i$-form $G$-symmetry is implemented by a functor $F:\rB^{i+1} G\to\npVect$.  We obtain the following equivalence of categories in arbitrary dimension, which can be interpreted as an equivalence of theories equipped with 0-form $G$-symmetry, and theories relative to the $\nRep(G)$ theory defined above.  
\begin{proposition}\label{prop:equiv}
  Let $G$ be an affine algebraic group over characteristic $0$ field.  Equivariantization gives an equivalence of categories $$\npRep(G):=\Fun(\rB G, \underline{\npVect})\simeq \Mod_{\nRep(G)},$$
  where $\underline{\npVect}$ is the geometric category associated to $\npVect$. 
\end{proposition}
This is a consequence of \Cref{lem:affinetomod}.  The trivial $(n+1)$d theory is identified with the monoidal unit $\nVect\in \npVect$, endowed trivial action of $G$.  Under the above equivalence, equivariantizing the $G$-action of the trivial theory results in $\nRep(G)$, which should be interpreted as a pure gauge theory.

\iffalse
pick the category of $(n+1)$d theories to be $\npVect$\footnote{We pick this category not because it is \textit{the} category of quantum field theories, but because the cobordism hypothesis allows us to use these categories to construct TQFTs. Thus they are a useful tool in the procedure that we outline to recover symmetries.}.  
\begin{notation}
   Let $\mathcal{T}$ be a $(n+1)$-dimensional QFT  with an $i$-form $G$-symmetry, where $G$ is an affine algebraic group. Then we say that the categorical symmetry of $\cT$ is given by a functor 
\end{notation}
\fi

%by viewing a theory as an object in the $(n+1)$-category of all $(n+1)$d theories.  For our purposes, the category of $(n+1)$d theories is taken to be $\npVect$, where $\npVect$ is defined as the geometric $(\infty,n)$-category whose objects are quasi-coherent sheaves of $(n-1)$-categories on $\Spec(\mathbb{K})$.  

Consider a subgroup $H\subset G$, which we interpret as the \emph{unbroken subgroup} for $G$.  The boundary determined by $\nRep(H)$ implements symmetry breaking on a physical boundary after compactifying.  The symmetry on the $\nRep(H)$ boundary is given by the category of $\nRep(G)$-module endofunctors of $\nRep(H)$.  The following theorems encapsulate the electric and magnetic symmetries of $G$-gauge theory in $(n+1)$-dimensions with unbroken subgroup $H$, supporting \Cref{claim:mainclaim}.

\begin{theorem}\label{mainthm:Wilson}
   Suppose that $G$ is an affine algebraic group over characteristic $0$ field, and $H\subset G$ is a closed subgroup.  There is an equivalence of monoidal categories 
   \begin{equation}
       \Fun^\mathrm{L}_{\nRep(G)}(\nRep(H),\nRep(H))\simeq \nQcoh^\star(H\backslash G/H). 
   \end{equation}
   In particular, there is a functor $\rB^{n-1}\Rep(H)\to \End^{\mathrm{L}}_{\nRep(G)}(\nRep(H))$. 
\end{theorem}

\begin{theorem}\label{prop:linkingelectric}
Suppose that G is an affine algebraic group over characteristic $0$ field, and let $H \subset G$ be a closed subgroup.  Consider the canonical functor 
\begin{equation}
    \FZ(\nQcoh^\star(H\backslash G/H))\xrightarrow{\varpi} \nQcoh^\star(H\backslash G/H). 
\end{equation}
There is a faithful monoidal functor $\rZ(H)\to \fib(\varpi)$.
\end{theorem}
Here, $H\backslash G/H:=\rB H\times_{\rB G}\rB H$ denotes the (stacky) \emph{double coset}.  \Cref{mainthm:Wilson} is realized as a consequence of \Cref{corr:GroupMorita} and \Cref{lemma:thmBpart1}, and \Cref{prop:linkingelectric} is proven in \Cref{thm:fiber}.  

  \Cref{mainthm:Wilson} can be interpreted as picking out a family of $1$-dimensional $\Rep(H)$-operators arising from laying $1$-dimensional bulk operators on the boundary, which become \emph{Wilson lines}.
\Cref{prop:linkingelectric} can be interpreted as picking out a family of codimension-$2$ $\rZ(H)$ operators, coming from terminating the operators which link with the 1-dimensional bulk operators.

The $\rZ(H)$ operators are consistent with symmetry breaking  by matter charged in a $G$ representation where $H$ acts trivially: they correspond to the remaining unbroken part of the 1-form symmetry.   %Recovery of these operators agrees with the discrepancy between flat and dynamical $G$-gauge theories, recognized in \cite{Bonetti:2024cjk}, whereby symmetries of flat gauge theories are implemented by Dirichlet boundary conditions, while symmetries in dynamical gauge theories are implemented by Neumann boundary conditions.   

%In the case that $H\subset G$ is central, the \textit{stack} double coset decomposes as $H\backslash G/H\simeq G/H\times \rB H$, and so one identifies the $1$-dimensional $\Qcoh^\star(H\backslash G/H)\simeq \Qcoh(G/H)\otimes_{\Vect}\Rep(H)$-symmetry operators with topological vacuum degeneracy and Wilson line operators.  For arbitrary (not necessarily central) subgroup $H\subset G$, the double coset $H\backslash G /H$ has a decomposition into a product over coset/centralizer components, giving a decomposition of $\Qcoh^\star(H\backslash G/H)$ into products of degeneracy and Wilson line parts. 

We now relate modules that implement 1-form symmetry breaking symmetry breaking to those which leave the symmetry unbroken.  This theorem can be interpreted as giving an equivalence of bulk SymTFTs for each of the boundary theories presented above.  This results in a uniform gauging procedure for theories with different charged matter.

\begin{maintheorem}[\Cref{thm:MoritaEquivalence}]\label{thm:moritaequivforn}
       Let $G$ be an affine algebraic group of finite type, and $H\subset G$ a closed subgroup.  Suppose that \underline{\textbf{either}}: $n=1$ and $H,G$ are reductive, \underline{\textbf{or}} $n\geq 2$.  Then there is a Morita equivalence 
   \begin{equation}
       \Mod_{\nRep(G)}\xrightarrow{\,\,\sim\,\,}\Mod_{\nQcoh^\star(H\backslash G/H)}\,. 
   \end{equation}
  As a consequence, we have an equivalence of Drinfeld centers
 \begin{equation}
\FZ(\nRep(G))\simeq \FZ(\nQcoh^\star(H\backslash G/H)). 
 \end{equation} 
\end{maintheorem}

\noindent \Cref{thm:MoritaEquivalence} provides a more general result for any $n$-affine stack, and \Cref{thm:moritaequivforn} is a straightforward consequence of this. The Drinfeld center also admits a description in terms of the (stacky) conjugation quotient $G/G=\rB G\times_{\rB G\times \rB G} \rB G$ as in \cite[Theorem 1.7]{ben2010integral}:
\begin{maintheorem}[{\Cref{prop:center}}]
    Let $G$ be an affine algebraic group over characteristic $0$ field, and $H\subset G$ a closed subgroup. There is an equivalence of braided categories
   \begin{equation}
       \FZ(\nRep(G)) \simeq \FZ(\nQcoh^\star(H\backslash G/H))\simeq \nQcoh^\star(G/G). 
   \end{equation}
\end{maintheorem}
This gives the expected decomposition of the Drinfeld center in terms of conjugacy classes and centralizers (see \Cref{eqn:DrinfeldDecomp}).  The following classification result, combined with our gauging procedure, recovers the expected magnetic symmetries of $G$-gauge theory. 

 \begin{proposition}
     Let $H$ be an affine algebraic group.  Every invertible $\Rep(H)$-module category is given by $\Rep^\alpha(H)$ where $
     \alpha$ is a class in $\rH^2(\rB H; \mathbb{G}_m)$.
 \end{proposition}
\noindent This gives a functor $\rH^2(\rB H; \mathbb{G}_m)\rightarrow \Mod_{\Rep(G)}$ which takes a class $\alpha$ to $\Rep^\alpha(H)$.  This can be viewed as a $(-1)$-form magnetic symmetry for (1+1)d theories.

For a general $(n+1)$-dimensional theory that is the boundary of a  $(n+2)$-dimensional SymTFT, we find the magnetic symmetry is still implemented by a codimension-2 operator by the following result:
\begin{theorem}\label{thm:electricmagnetic}
There is a functor which implements the magnetic symmetry of the (n+1)d boundary determined by $\nRep(H)$, of the TQFT constructed in \Cref{mainthm:tqft},  which takes the form
\begin{align}
\rB^{n-1}(\rH^2(\rB H;\mathbb{G}_m))\to \Mod_{\nRep(G)}&&\alpha\mapsto \Rep^\alpha(H)
\end{align}
\end{theorem}
We describe the relationship between $\rH^2(\rB H; \mathbb{G}_m)$ and magnetic symmetries in dynamical gauge theories in \S\ref{subsection:2dtwistedsectors}. We note that the magnetic symmetry is trivially present in flat $G$-gauge theory, as it is obtained by the integral of the curvature, which is trivial when the background connection is flat. In dynamical gauge theories they can act non-trivially on 't Hooft operators, see \S\ref{subsection:symgauge}.

%%%%%%%%%%%%%%%%%%%%%%%%%%%%%%%%%%
\subsection{Outline}\label{subsection:outline}
%%%%%%%%%%%%%%%%%%%%%%%%%%%%%%%%%%
We dedicate \S\ref{section:symTFT} to recovering symmetries of (1+1)d $G$-gauge theories.  Technical details are offloaded to later sections. In \S\ref{subsection:symgauge} we review the symmetries of standard gauge theories.  In \S\ref{subsection:buldTFT} we construct a (2+1)d SymTFT, and use it to recover symmetries in \S\ref{subsection:2dsymmetries}.
In \S\ref{section:geometry} we introduce algebraic geometry background.  In \S\ref{section:dynamicalgauging} we prove our main results regarding gauging in general dimensions. We construct a higher dimensional SymTFT in \S\ref{subsection:generalTQFT} and identify electric symmetries in \S\ref{subsection:extractingelectric}, and magnetic symmetries in \S\ref{subsection:twistedsectors}.  Proofs of technical resuts are postponed until \S\ref{section:technicalproofs}.

%%%%%%%%%%%%%%%%%%%%%%%%%%%%%%%%%%
\subsection{Glossary}\label{subsection:glossary}
%%%%%%%%%%%%%%%%%%%%%%%%%%%%%%%%%%
We usually write categories in bold, to distinguish them from spaces and stacks which are written in normal font.
\begin{itemize}
\item Denote the multiplicative group by $\mathbb{G}_m$.  Over $\mathbb{C}$, we have the identification $\mathbb{G}_m\simeq \mathbb{C}^\times$.
\item Denote the group of \emph{n-th roots of unity} by $\mu_n\subset \mathbb{G}_m$. 
\item An \textit{affine algebraic group} is a group-scheme $G$ such that the underlying scheme is affine. 
\item An algebraic group $G$ is \textit{reductive} if every finite-dimensional representation of $G$ is completely reducible. Examples of reductive affine groups include: finite groups, $\mathbb{G}_m$, $\mathrm{GL}_n$, $\mathrm{SL}_n$, $\mathrm{SO}_n$, $\mathrm{Sp}_n$, exceptional Lie groups, tori. 
\item $\Spaces$: The $(\infty,1)$-category of small $\infty$-groupoids.
\item $\CAlg$: The $(\infty,1)$-category of (spectral) commutative algebras. 
\item $\Mod_A(\cC)$: The category of $A$-modules in $\cC$.
\item $\Aff/\Stk$: The $(\infty,1)$-categories of affine schemes, and stacks, respectively. 
\item $\Qcoh^*(X)/\Qcoh^\otimes(X)$: The $(\infty,1)$-categories of quasi-coherent sheaves on stack $X$ with convolution, and symmetric monoidal product, respectively.
     \item $\Pr^{\mathrm{L}}/\ \Pr^{\mathrm{L}}_\mathrm{St}$: the $(\infty,2)$-category of presentable/presentable-stable categories and colimit-preserving functors.
\end{itemize}

%%%%%%%%%%%%%%%%%%%%%%%%%%%%%%%%%%
\section{Warm-up: (1+1)-Dimensional Gauging and Symmetries}\label{section:symTFT}
%%%%%%%%%%%%%%%%%%%%%%%%%%%%%%%%%%
In this section, we consider the case of (1+1)d gauge theories and study how to recover their symmetries.  We relate invertible $\Rep(G)$-module categories to magnetic symmetries, and endomorphisms of module categories to broken electric symmetries for matter transforming in $G$-representations.

We will also use this section to introduce some background material which will be useful in later sections.  To guide the reader, we give an overview of the results pertaining specifically to (1+1)d gauging:
\begin{itemize}
    \item Using $\Rep(G)$, we construct a once-categorified (1+1)d SymTFT valued in the Morita 3-category $\Mor_1(\Pr^\mathrm{L}_{\st})$\, in \Cref{prop:buildTQFT}. 
    \item We construct invertible modules of $\Rep(G)$, which lead to magnetic $(-1)$-form $\rH^2(\rB G;\mathbb{C}^\times)$  symmetry.
    \item We recover topological Wilson line operators in \Cref{lemma:symflat} and 1-form $\rZ(G)$-symmetry in \Cref{cor:endingelectric}, and explain the relationship with the usual magnetic symmetry exhibited in flat $G$ gauge theory. 
    %Using the $n=1$ case of \Cref{thm:fiber} which we prove in \S\ref{subsection:electricsym},  
    \item We show how to construct a module with both electric and magnetic symmetry in \Cref{prop:2dmagnetic}. 
\end{itemize}

%%%%%%%%%%%%%%%%%%%%%%%%%%%%%%%%%%
\subsection{Symmetries of Gauge Theories}\label{subsection:symgauge}
%%%%%%%%%%%%%%%%%%%%%%%%%%%%%%%%%%

We first discuss symmetries of a flat $G$-gauge theory, obtained from the empty theory by summing over only flat connections in the path integral.  This does not allow local propagating degrees of freedom.  Flat gauging of a global symmetry can be implemented by inserting a network of topological operators \cite{Frohlich:2009gb,Bhardwaj:2017xup}.  This procedure can not recreate non-flat connections, and so there is no analogous dynamical gauging procedure.  

Since multiple equivalent presentations of flat gauge theory exist, we refrain from selecting one particular formulation and instead emphasize several illustrative examples in low dimensions
which include: the (2+1)d and (3+1)d Toric code \cite{Kong:2019brm,Kong:2020wmn,Delcamp:2020rds,Johnson-Freyd:2020twl}, (anomalous) topological order in (3+1)d \cite{Decoppet:2025eic}, and BF theory and certain string-net models \cite{Levin:2004mi}.
Each of these theories share an important collection of symmetry operators, given by:
\begin{itemize}
    \item codimension-$2$ symmetries implemented by Gukov-Witten operators, labeled by conjugacy classes of $G$.
    \item $1$-dimensional topological Wilson lines labeled by representations of $G$, (which are $(d-2)$-form symmetries).  
\end{itemize}
These two symmetries link with each other in the case of flat $G$-gauge theory, where the topological Wilson lines should be thought of as  the generator of dual symmetry from gauging $G$, and the Gukov-Witten operators generate the electric symmetry.

\begin{rem}\label{rem:ondynamical}
  While dynamical gauge theories are not our main focus, we remark that the symmetries of a dynamical theory are a subset of the symmetries in flat $G$-gauge theory.  To capture dynamical gauging, it would also be necessary to encode local data. 
  
  We start with a review of the symmetries of $\rU(1)$ Maxwell theory on a $d$-dimensional manifold $\mathcal{M}_d$.  Pure $\rU(1)$ Maxwell theory is a gauge theory obtained by dynamically gauging a 0-form $\rU(1)$ symmetry on the vacuum (i.e. summing over all connections in the path integral).  The action is given by\footnote{In general there is freedom to add a  topological $\theta$-term into the Lagrangian. Such topological terms do not affect the equations of motion. We will not consider such terms for our discussion.} 
\begin{equation}
    S = \frac{1}{2e^2 }\int_{\mathcal{M}_d} \,F \wedge \star F\,.
\end{equation}
where $e \in \mathbb{R}^+$ is the gauge coupling, and $F$ is the field strength of a $\rU(1)$ connection $A$.  The equation of motion and Bianchi identity induce two symmetries:
\begin{itemize}
    \item An ``electric'' 1-form $\rU(1)$ symmetry generated by $\exp\left(i \frac{\alpha}{e^2} \int_X \star F\right)$, where $\alpha \in \rU(1)$, acting on  Wilson line operators $\exp\left(i n \int_\gamma A\right)$ where $n \in \mathbb Z$ and $\gamma$ is a curve. 
    \item A ``magnetic'' $(d-3)$-form $\rU(1)$ symmetry generated by $\exp\left(i \frac{\alpha}{2\pi} \int_{\cM_2} F\right)$, acting on  `t Hooft lines \footnote{In (1+1)d the 't Hooft lines are $(-1)$-dimensional and hence not present; nevertheless the magnetic symmetry still exists as a $(-1)$-form in the theory.}. 
\end{itemize}
    More generally, for a $G$ Yang--Mills theory in $d$-dimensions, one has 
    \begin{itemize}
        \item $1$-form electric $\rZ(G)$ and dimension-$1$ electrically charged Wilson lines labeled by $\Rep(G)$.  
        \item $(d-3)$-form magnetic $\widehat{\pi_1(G)}$ and dimension-$(d-3)$ magnetically charged `t Hooft operators labeled by $\pi_1(G)$, where $\widehat{\pi_1(G)}$ is the Pontryagin dual. 
    \end{itemize}
See \cite{Bhardwaj:2023kri,Schafer-Nameki:2023jdn} for a more detailed review. It is important to note that the magnetic symmetry in the above case acts nontrivially in dynamical theories, due to the presence of monopole operators. This is not the case in flat $G$-gauge theory; while the magnetic symmetry is the same, it acts trivially in this case.
One can consider adding in matter to pure Yang--Mills, which transforms in some representation $R$ of the gauge group $G$. If $R$ transforms  non-trivially under the electric symmetry, the matter will break the $\mathrm{Z}(G)$ symmetry to a subgroup $\mathrm{Z}(H)\subset \mathrm{Z}(G)$ that acts trivially on the representation $R$.
\end{rem}

Recent studies of continuous gauging have been conducted from the perspective of SymTFT.  The group $\rU(1)$ is considerd in \cite{Antinucci:2024zjp,Brennan:2024fgj,Arbalestrier:2025poq}.  The case of non-abelian continuous group was considered in \cite{Bonetti:2024cjk,Mignosa:2025cpg}.  The authors in \cite{Apruzzi:2024htg,Mignosa:2025cpg} use the ``SymTh'' approach to do dynamical gauging as it allows for a modified Neumann condition; the latter takes a more holographic perspective of this problem.  In \cite{Apruzzi:2025hvs} the SymTFT was generalized to be compatible with continuous spacetime symmetries. Several mathematical models were also discussed for continuous categorical symmetries, in particular \cite{Freed:2009qp,Jia:2025vrj} discusses skyscraper sheaves, and their gauging in \cite{Jia:2025uun}. Weis \cite{Weis:2022egw} discusses manifolds tensor categories, which involve a smooth family of simple objects. These models do not allow for objects
with continuous support, as only sheaves with support were considered. To ameliorate this \cite{Marin:2025stc}  proposes the use of continuous tensor categories defined from representations of $\mathrm{C}^*$-algebras.

%%%%%%%%%%%%%%%%%%%%%%%%%%%%%%%%%%%
\subsection{Constructing the (2+1)d SymTFT}\label{subsection:buldTFT}
%%%%%%%%%%%%%%%%%%%%%%%%%%%%%%%%%

We are now ready to build a once-categorified (1+1)d TQFT from $\Rep(G)$, which was obtained by equivariantizing the trivial symmetry on $\Vect$.  By once-categorified, we mean that this TQFT assigns to a 2-dimensional manifold a 2-morphism in the $3$-category $\Mor_1(\Pr^\mathrm{L}_\st)$.  In the case that $G$ is finite, this TQFT can be extended to a (2+1)d TQFT.  One should view such a once-categorified theory as a (2+1)-dimensional TQFT which is only defined on a subclass of three-manifolds.

We must first define the Morita category associated to an $(\infty,1)$-category $\cC$ as in \cite{HaugsengMorita,JFS2017}.  This will be the target category for our functorial TQFT.  

\begin{definition}
Let $\cC$ be a monoidal $(\infty,n)$-category\footnote{With good relative tensor products.  This is necessary to define composition of morphisms.}.  Denote by $\Mor(\cC)$ the \emph{Morita category} of $\cC$:  The $(\infty,n+1)$-category whose objects are algebras in $\cC$, whose 1-morphisms given by bimodule objects in $\cC$, and $i$-morphisms for $i\geq 2$ are given by bimodule $(i-1)$-morphisms. 
\end{definition}

\begin{proposition}\label{prop:buildTQFT}
     Let $G$ be an affine algebraic group over characteristic $0$ field.  The category $\Rep(G)$ is 2-dualizable in $\Mor_1(\Pr^\mathrm{L}_\st)$, and so can be used to construct a once-categorified fully-extended (1+1)-dimensional TQFT valued in $\Mor_1(\Pr^\mathrm{L}_\st)$.
\end{proposition}
\begin{proof}
     By \cite[Corollary 4.1.6]{stefanich2023}, $\Rep(G)$ is smooth and proper as an algebra in $\Mor_1(\Pr^\mathrm{L}_{\mathrm{St}})$. By \cite[Proposition 4.6.4.12]{LurieHA}, smooth implies that the coevaluation module is right dualizable, equivalently, if the evaluation module is left dualizable. By \cite[Proposition 4.6.4.4,]{LurieHA}, proper implies that coevaluation module is left dualizable, equivalently, the evaluation module is right dualizable.
       Hence, $\Rep(G)$ is $2$-dualizible in $\Mor_1(\Pr^{\mathrm{L}}_{\mathrm{st}})$, which is itself an $(\infty,3)$-category.  By the cobordism hypothesis \cite{LurieTFT}, this defines a once-categorified fully-extended framed (1+1)-dimensional TQFT valued in $\Mor_1(\Pr^{\mathrm{L}}_{\mathrm{st}})$, whose value on a point is $\Rep(G)$. 
\end{proof}

\begin{rem}
   A TQFT of the kind defined in \Cref{prop:buildTQFT}, and higher dimensional defects for it, were constructed in \cite[Section 4.2]{Kinnear2024} and called the ``classical $G$-gauge theory''. 
\end{rem}

By cobordism hypothesis with singularities \cite[Theorem 4.3.11]{LurieTFT}, any dualizable $\Rep(G)$-module category defines a (1+1)d boundary for the (2+1)d bulk theory defined in \Cref{prop:buildTQFT}.  For every closed subgroup $H\subset G$ and $\alpha\in \H^2(\rB G, \mathbb{C}^\times)$, there is a dualizable module category $\Rep^\alpha(H)$ for $\Rep(G)$ by \Cref{prop:moritaforn}.  This gives a (1+1)d boundary for every $H$ and $\alpha$.  This is depicted in \Cref{fig:2dSymTFT}, right.  See also \cite{Cattaneo:1996pz,Cattaneo:2000mc,Cattaneo:2002tk,Jia:2025jmn} for a BF theory treatment of the TQFT. 

\begin{figure}
    \centering
\begin{tikzpicture}[thick]
  % --- LEFT BLUE PLATE (Vacuum/Vect) ---
  \fill[shade, top color=blue!65, bottom color=blue!8, shading angle=180, opacity=1.2]
    (3,1) -- (8,1) -- (9,2) -- (4,2) -- cycle;
  \draw[black, dotted] (3,1) -- (8,1) -- (9,2) -- (4,2) -- cycle;

  \node at (6,1.5) {$\text{Vac}$};
  \node at (6,2.5) {$\mathbf{Vect}$};

  % Arrow underneath with G label
  \draw[<-,decorate] (6,.9) .. controls (7.25,-.25) and (4.2,-.25) .. (5.5,.9);
  \node at (5.8,-0.28) {$G$};

  % Arrow to right
  \draw[->,decorate,decoration={snake,amplitude=.4mm,segment length=2mm,post length=1mm}]
    (9.6,2.5) -- (11.6,2.5) node[midway, above] {$G$-Equiv.};

  % --- RIGHT LOWER GREEN PLATE ---
  \fill[shade, top color=green!65, bottom color=green!8, shading angle=-90, opacity=1.2]
    (13,1) -- (17,1) -- (18,2) -- (14,2) -- cycle;

  % --- RIGHT UPPER RED PLATE ---
  \fill[shade, top color=red!65, bottom color=red!8, shading angle=-90, opacity=1.2]
    (13,3.5) -- (17,3.5) -- (18,4.5) -- (14,4.5) -- cycle;

  % Outlines for right stack
  %\draw[black,dotted] (12,1)--(13,1);
  \draw[black] (13,1)-- (17,1);
  %\draw[black,dotted] (13,2)--(14,2);
  \draw[black] (14,2)-- (18,2);
  %\draw[black] (17,1)--(18,2);
 % \draw[black,dotted] (12,3.5) -- (13,3.5);
  \draw[black] (13,3.5) -- (17,3.5);
  %\draw[black,dotted] (13,4.5) -- (14,4.5);
  \draw[black] (14,4.5) -- (18,4.5);
  \draw[white] (17,3.5)--(18,4.5);
  \draw[white] (17,1)--(18,2);
  %\draw[black] (17,1) -- (17,3.5)--(18,4.5)--(18,2);

  % Labels
  \node at (15.5,1.5) {$\mathfrak{B}_{\mathrm{Phys}}$};
  \node at (15.5,2.7) {$\Rep(G)$};
  \node at (15.5,4) {$\mathfrak{B}_{\text{Top}}$};
\end{tikzpicture}
%\dev{The boundary only the left bulk has to be trivial.  Right side is also noncompact}
    \caption{Left: The (1+1)d empty theory is equipped with the trivial action of a group $G$.  Right: The bulk-boundary system obtained from equivariantizing the $G$-action. 
    The bulk contains the (2+1)d TQFT constructed from $\Rep(G)$.  The physical boundary $\mathfrak{B}_{\mathrm{Phys}}$ contains only local degrees of freedom.  The topological boundary $\mathfrak{B}_{\mathrm{Top}}$ contains only topological degrees of freedom, and corresponds to a $\Rep(G)$-module category.  The TQFT is a once-categorified (1+1)d theory, and we denote this by not capping off the ends in the 3d bulk. A detailed physical account of this ``non-compact'' TQFT, in the sense that the objects are non-compact, is descibed in \cite[Appendix B]{Antinucci:2024bcm}.
    %Boundaries are dotted to indicate a non-compact space.
    }
    \label{fig:2dSymTFT}
\end{figure}
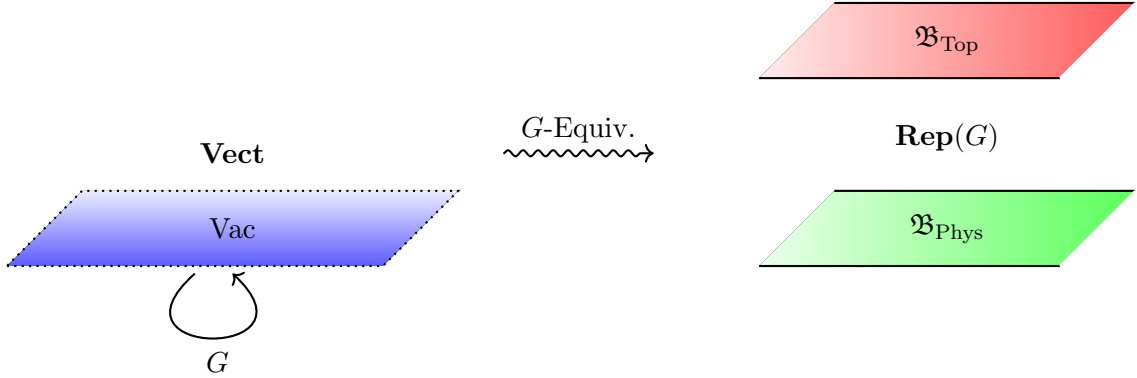

%%%%%%%%%%%%%%%%%%%%%%%%%%%%%%%%%%
\subsection{Symmetries of Gauge Theories in (1+1)d}\label{subsection:2dsymmetries}
%%%%%%%%%%%%%%%%%%%%%%%%%%%%%%%%%%
The primary goal of this section is to present concrete physical applications of our machinery without delving into more abstract generalities.  For this reason, we identify the electric and magnetic symmetries of (1+1)d gauge theories but  will choose to call upon to the results and background from \S\ref{subsection:extractingelectric} and \S\ref{section:geometry} as needed. 

%%%%%%%%%%%%%%%%%%%%%%%%%%%%%%%%%%
\subsubsection{Identifying  the Electric Symmetry
}\label{subsection:electricsym2d}
%%%%%%%%%%%%%%%%%%%%%%%%%%%%%%%%%%
We begin by identifying the electric symmetry that appears in flat $G$-gauge theory, which is captured by the following lemma. 

\begin{lemma}\label{lemma:symflat}
Let $H\subset G$ be a closed subgroup and $\alpha\in \H^2(\rB H;\mathbb{C}^\times)$.  The (1+1)d boundary theory determined by $\Rep^\alpha(H)$ exhibits a 0-form $\Rep(H)$-symmetry.
\end{lemma}
%\begin{proof}
 %    Thanks to \Cref{eq:doublecoset}, we see that   $\End_{\Rep(G)}(\Rep(H)) = \Qcoh^\star(H\backslash G/H)$. There is a monoidal functor $\Qcoh^\star(H\backslash G/H) \rightarrow \Omega_{\Rep(H)} \Mod_{\Rep(G)}$ which deloops to a functor
%\begin{equation}\label{eq:0form}
 %   \rB \Qcoh^\star(H\backslash G/H) \rightarrow  \Mod_{\Rep(G)}\,.
%\end{equation}
 %In the case of $H=G$, then $\Qcoh^\star(H\backslash G/H) = \Qcoh(\rB G) = \Rep(G)$ is the symmetry one expects for flat $G$-gauge theory. If $H$ is central in $G$ then by \Cref{eq:doublecoset} we get 
 %\begin{equation}
%\Qcoh^\star(H\backslash G/H)\simeq \Qcoh^\star(G/H\times \rB H)\simeq \Qcoh(G/H)\otimes_{\Vect}\Rep(H).
%\end{equation}
%Therefore, the analog of the functor in \Cref{eq:0form} becomes
%\begin{equation}
 %   \rB\Qcoh(G/H)\otimes_{\Vect}\rB \Rep(H)\simeq \rB \Qcoh^\star(H\backslash G/H)\to \Mod_{\Rep(G)}, 
%\end{equation}
%where here, $G/H$ is the ordinary quotient, not the stack quotient.
%\end{proof}
\noindent The contents of this Lemma is a special case of \Cref{lemma:thmBpart1}, and the contents will be spelled out in the proof therein.

We now discuss how to identify the center 1-form electric symmetry. 
The lines in the bulk TQFT are given by the Drinfeld center $\FZ(\Rep(G))$ (see \Cref{def:drinfeld}).  We will now determine which bulk lines can terminate on the $\Rep(H)$-boundary.  These are the lines that are permitted to end on the boundary, and thus label end points of the blue lines in \Cref{fig:ending}, which are codimension-2 boundary operators.

In the case where the bulk TQFT is constructed from a fusion category $\cC$, a (1+1)d boundary is determined by a commutative algebra object $\mathcal{A} \in \FZ(\cC)$.  The category $\Bimod_{\cA|\cA}(\cC)$ describes bulk lines that satisfy a Dirichlet boundary condition on top boundary of \Cref{fig:ending}.  The lines permitted to end on the boundary are those contained in the algebra object, and as a result, those that contain the vacuum line as a summand on the boundary.  The `ending' operator on the boundary is given by the projection to the vacuum line.  The symmetry breaking in \Cref{lemma:symflat} has been studied in the case of finite $G$, and used to classify gapped/gapless phases in (1+1)d and higher (2+1)d 
\cite{Bhardwaj:2023idu,Bhardwaj:2023fca,Bhardwaj:2023bbf,Bhardwaj:2024qrf,Bhardwaj:2024qiv,Bhardwaj:2025piv,Wen:2023otf,Wen:2024qsg,Wen:2025thg}.     
      
This procedure is difficult to replicate in the continuous case.  Instead, we notice that if a line is trivialized on the boundary, then it is certainly endable\footnote{This is not a necessary condition: a line may be endable, in the sense that it has a morphism to the vacuum line, even if it is not trivialized by the bulk-boundary map.}.  Lines which are trivialized on the boundary are equivalently those that appear in the fiber of the map $\FZ(\cC)\to \Bimod_{\cA|\cA}(\cC)$.  Then we instead proceed by showing that the expected broken electric center $1$-form symmetry appears in the fiber of $\FZ(\Rep(G))\rightarrow \End_{\Rep(G)}(\Rep(H))$, induced by the Morita equivalence of $\cC$ and $\Bimod_{\cA|\cA}(\cC)$.  This result follows from \Cref{thm:fiber}.  This agrees with the $1$-form symmetry breaking that flat $G$-gauge theory.

\begin{corollary}\label{cor:endingelectric}
$Z(H)$ is contained faithfully within the bulk lines of $\FZ(\Rep(G))$ that may be ended on the boundary determined by $\Rep(H)$. 
\end{corollary}

 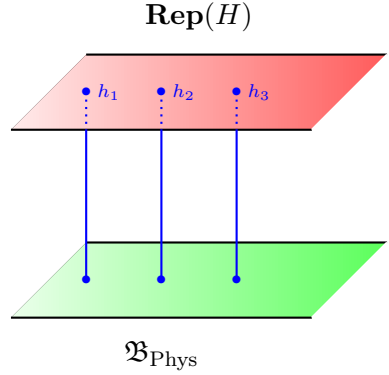
\begin{figure}
    \centering
\begin{tikzpicture}[thick]
  % bottom plate: vertical gradient (darker near the plate, lighter away)
  \fill[shade, top color=green!65, bottom color=green!8, shading angle=-90, opacity=1.2]
    (13,1) -- (17,1) -- (18,2) -- (14,2) -- cycle;

  % --- RIGHT UPPER RED PLATE ---
  \fill[shade, top color=red!65, bottom color=red!8, shading angle=-90, opacity=1.2]
    (13,3.5) -- (17,3.5) -- (18,4.5) -- (14,4.5) -- cycle;

  % Outlines for right stack
  %\draw[black,dotted] (12,1)--(13,1);
  \draw[black] (13,1)-- (17,1);
  %\draw[black,dotted] (13,2)--(14,2);
  \draw[black] (14,2)-- (18,2);
  %\draw[black] (17,1)--(18,2);
  %\draw[black,dotted] (12,3.5) -- (13,3.5);
  \draw[black] (13,3.5) -- (17,3.5);
  %\draw[black,dotted] (13,4.5) -- (14,4.5);
  \draw[black] (14,4.5) -- (18,4.5);
  \draw[white] (17,3.5)--(18,4.5);
  \draw[white] (17,1)--(18,2);
  %\draw[black] (17,1) -- (17,3.5)--(18,4.5)--(18,2);

  % Labels
  \node at (15.5,5) {$\Rep(H)$};
 % \node at (15.5,2.7) {$\Rep(G)$};
  %\node at (15.5,4) {$\mathfrak{B}_{\text{Top}}$};
  \draw[blue,dotted] (16,4)-- (16,3.5);
  \draw[blue] (16,3.5)-- (16,1.5);
  \draw[blue,dotted] (15,4)-- (15,3.5);
  \draw[blue] (15,3.5)-- (15,1.5);
  \draw[blue,dotted] (14,4)-- (14,3.5);
  \draw[blue] (14,3.5)-- (14,1.5);

  \node[blue, font=\tiny] at (16,4) {$\bullet$};
  \node[blue, font=\tiny] at (15,4) {$\bullet$};
  \node[blue,font=\tiny] at (14,4) {$\bullet$};
  \node[blue,font=\tiny] at (14.3,4) {$h_1$};
  \node[blue, font=\tiny] at (15.3,4) {$h_2$};
  \node[blue,font=\tiny] at (16.3,4) {$h_3$};

  \node[blue,font=\tiny] at (15,1.5) {$\bullet$};
  \node[blue,font=\tiny] at (14,1.5) {$\bullet$};
  \node[blue, font=\tiny] at (16,1.5) {$\bullet$};

  \node at (15,.5) {$\mathfrak{B}_{\mathrm{Phys}}$};
\end{tikzpicture}
    \caption{$\Rep(H)$ determines a topological boundary for the bulk $\Rep(G)$ TQFT.  The blue lines are objects in $\FZ(\Rep(G))$ which can end on the topological boundary. The end points  carry a label $h_i\in \rZ(H)$.}
    \label{fig:ending}
\end{figure}

%%%%%%%%%%%%%%%%%%%%%%%%%%%%%%%%%%
\subsubsection{Identifying the Magnetic Symmetry}\label{subsection:2dtwistedsectors}
%%%%%%%%%%%%%%%%%%%%%%%%%%%%%%%%%%
  
  The magnetic symmetry referred to in this section is the one discussed in \Cref{rem:ondynamical} in the context of flat $G$-gauge theory.  This is in contrast to the terminology used in \cite[Corollary 2.3.3]{Johnson-Freyd:2020twl}  \cite{Bhardwaj:2023wzd}, where magnetic symmetry refers to the $1$-form operators linking with $\Rep(H)$ (which we called electric in the previous section).  
  
  %for the terminology and for a general discussion of symmetry and charges.

 % The magnetic symmetry in a flat \(G\)- theory acts naturally on the 't~Hooft line operators, realized as the worldlines of probe magnetic monopoles, resulting in a $(-1)$-form symmetry in (1+1)d. \matt{we should say there is also a magnetic symmetry in flat $G$ gauge theory, but its trivial.}
  
  We show how such a symmetry can appear from our gauging procedure. Specifically, we show that a single $\Rep(G)$-module can exhibit both the 1-form $\rZ(H)$-symmetry and $(-1)$-form magnetic symmetry. To do this, we require the following lemma, whose details are given in \Cref{contruction:invertible}.
\begin{lemma}\label{lemma:invertiblemodule}
       Invertible $\Rep(G)$-modules are classified by classes in $\rH^2(\rB G; \mathbb{C}^\times)$. 
  \end{lemma}
Using this result we can construct a module which exhibit both electric and magnetic symmetry.  One can show that $\Fun_{\Rep(G)}(\Rep^\alpha(H),\Rep^\beta(H))$ contains $\Rep^{\beta-\alpha}(H)$ in \Cref{ex:dyonic}.  We may interpret these operators as dyonic, with magnetic charge $\beta-\alpha$.  In particular, when $\beta-\alpha$ is trivial, we recover the usual Wilson line operators, which act as operators $\Rep^\alpha(H)\to \Rep^\alpha(H)$ preserving magnetic sector.  The following proves the (1+1)d case of \Cref{thm:electricmagnetic}. 
\begin{proposition}\label{prop:2dmagnetic}
    The module $\Rep(H)\in \Mod_{\Rep(G)}$, can be extended to a module with both $(-1)$-form $\rH^2(\rB H; \mathbb{C}^\times)$ magnetic symmetry and codimension-1 dyonic $\Rep^\alpha(H)$ operators.
\end{proposition}
  \begin{proof}
%Let $H$ in \Cref{cor:endingelectric}. Then the objects that are trivialized under the map $\FZ(\Rep(G))\rightarrow \End_{\Rep(G)}(\Rep(H))$ are labeled by $\rZ(H)$. Hence the module $\Rep(H)$ contains codimesion-2 operators labeled by $\rZ(H)$.  

Consider the following family of modules, exhibiting $(-1)$-form $\H^2(\rB H;\mathbb{C}^\times)$ symmetry, and broken electric symmetry by the preceeding discussion:
\begin{align}\label{eq:maglabels}
    \rH^2(\rB H;\mathbb{C}^\times)&\rightarrow \Mod_{\Rep(G)} \\ \notag
    \alpha &\mapsto \Rep^\alpha(H)\,.
\end{align}\end{proof}
%We note that the same principle of looking for Neumann boundary conditions of operators in the bulk TQFT also applies for the $(-1)$-form symmetry of dynamical gauge theories. In particular, they arise from picking the Neumann boundary condition for a family of SymTQFTs labeled by $\Rep^\alpha(G)$.

\begin{rem}\label{rem:MagneticYM}
     Starting from the classes in $\rH^2(\rB G; \mathbb{C}^\times)$ we now make contact with the description for the magnetic symmetry given in \Cref{rem:ondynamical}.  When $G$ is a connected group object in the category of manifolds, there is a short exact sequence
\begin{equation}\label{eq:cover}
    1 \longrightarrow \pi^{}_1(G) \longrightarrow \widetilde G \longrightarrow G \longrightarrow 1,
\end{equation}
where  $\widetilde G$ is the universal covering space of $G$.\footnote{In the algebraic case, technically $\widetilde{G}$ should be the finite étale covering space of $G$, and $\pi_1(G)$ becomes the étale fundamental group of $G$.} We get a long exact sequence in cohomology: 
\begin{equation}
     \ldots \rightarrow \rH^1(\rB \widetilde{G}; \mathbb{C}^\times)\rightarrow \rH^1(\rB \pi_1(G); \mathbb{C}^\times) \rightarrow \rH^2(\rB G; \mathbb{C}^\times)\rightarrow \rH^2(\rB \widetilde{G}; \mathbb{C}^\times) \rightarrow \ldots\,. 
\end{equation}
In the case where $G$ has no torus components, we get $\rH^1(\rB \widetilde{G}; \mathbb{C}^\times) = 0 =\rH^2(\rB \widetilde{G}; \mathbb{C}^\times)$, since $\widetilde{G}$ is the covering space of $G$.  This gives an equivalence 
\begin{equation}
    \widehat{\pi_1(G)}:=\rH^1(\rB \pi^{}_1(G); \mathbb{C}^\times) \xrightarrow{\sim} \rH^2(\rB G; \mathbb{C}^\times), 
\end{equation}
between $\rH^2(\rB G; \mathbb{C}^\times)$ and the Pontryagin dual of ${\pi}^{}_1(G)$. Based on \S\ref{subsection:symgauge}, this is why we dub the $\rH^2(\rB G; \mathbb{C}^\times)$-symmetry ``magentic''.

When $G$ has torus components, we do not have access to the above exact sequence in the algebraic world.  Consider the multiplicative group $\mathbb{C}^\times$ over $\mathbb{C}$.  In this case, one has a short exact sequence 
\begin{equation}
    1\to \pi_1(\mathbb{C}^\times)\to \mathbb{C}\xrightarrow{\mathrm{exp}} \mathbb{C}^\times\to 1\,.
\end{equation}
However, the covering map $\mathrm{exp}:\mathbb{C}\to \mathbb{C}^\times$ is \textbf{not} algebraic, since its expansion has terms in arbitrary degree.  We see that $\rH^1(\rB \widetilde{G};\mathbb{C}^\times) \neq 0$, and so we do not have an isomorphism between $\widehat{\pi_1(G)}$ and $\H^2(\rB G;\mathbb{C}^\times)$, which classifies invertible module categories. 

Here, we see the first failure of an algebro-geometric approach to dynamical gauging.  We could choose to replace algebraic stacks with \emph{analytic} or \emph{smooth} stacks, and perform a similar categorical analysis with categories internal to an analytic or smooth topos.  We will refrain in this work, since many of the geometric and higher-categorical results necessary for our constructions exist only in the algebro-geometric setting.  We intend to pursue the necessary generalizations in future work.

\end{rem}

%%%%%%%%%%%%%%%%%%%%%%%%%%%%%%%%%%
\section{Geometric Gadgets}\label{section:geometry}
%%%%%%%%%%%%%%%%%%%%%%%%%%%%%%%%%%
While we previously referred to these objects of derived algebraic geometry somewhat informally, a more precise treatment is required for the proofs of our main results.  We begin with some motivation.

Consider a finite group $G$.  Since $G$ is discrete, for each $g\in G$, there is a vector bundle supported only over $g$.  Consequently, there is an equivalence between $G$-graded vector spaces and vector bundles over $G$.  This is not true for continuous groups.  For example, there is a unique rank-1 vector bundle on the space $\mathbb{C}$, so one can not hope to capture the all properties of continuous symmetry using vector bundles.  
\emph{Skyscraper sheaves} are linear objects concentrated at discrete points.  There is a unique simple skyscraper sheaf associated to each point $g$ in a continuous group $G$, and these sheaves have been used to describe continuous symmetries in \cite{Freed:2009qp,Jia:2025vrj}.  However, one also expects operators which can be continuously supported on the entirety of the group $G$.  

\begin{figure}
\begin{tikzpicture}
  % shaded outer wall patch of the cylinder
  \shade[left color=red!10, right color=red!60]
     (0,0) arc (90:180:2.5cm and 0.5cm)
    -- (-2.5,-2)
    arc (180:90:2.5cm and 0.5cm)
    -- cycle;

  % top and bottom edge arcs for reference
   (0,0) arc (90:180:2.5cm and 0.5cm);
  %\draw[dashed] (0,-3) arc (90:180:1.5cm and 0.5cm);
  \draw[thick] (0,-2.5) ellipse (2.5cm and .5cm);
  \draw[blue,thick] (2.5,-2)--(2.5,-.5);
  \node[font = \tiny] at (-2.6,-2.7) {$\pi$};
  \node[font = \tiny] at (0,-2.2) {$\frac{\pi}{2}$};
  \node[font = \tiny] at (2.6,-2.7) {$0$};
\end{tikzpicture}
    \caption{Sheaves on $S^1$: a skyscraper supported at $0$, and a sheaf supported on $[\frac{\pi}{2},\pi]$.}
    \label{fig:sheavesS1}
\end{figure}
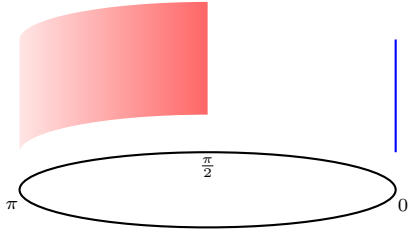
It is necessary to introduce the notion of \emph{quasi-coherent sheaves} which interpolates between these two extremes (see \Cref{fig:sheavesS1}).  The eager reader will be well served by reading only \S\ref{subection:geometry}, \S\ref{subsection:QCoh}, and \S\ref{subsection:Convolution} below, and then skipping to \S \ref{section:dynamicalgauging}.  We refer the reader to \cite{GaitsgoryRozenblyum, LurieSAG} for more complete treatments.

\subsection{Presheaves and Prestacks}\label{subection:geometry}
%%%%%%%%%%%%%%%%%%%%%%%%%%%%%%%%%
Let $\Sch$ be a category of geometric spaces,\footnote{The notation $\Sch$ is short for ``scheme'', but the arguments presented here apply also for (analytic/smooth/topological/super) manifolds with an appropriate notion of commutative algebra.} and $\CAlg$ be the category of commutative algebras.  To any $X\in \Sch$, one associates the algebra of global functions $\Gamma(X)$, and functions can be pulled back along morphisms $Y\to X$ in $\Sch$.  One can define a ``global functions'' functor $\Gamma:\Sch\to \CAlg^{\mathrm{op}}$.

Some spaces $X$ can be reconstructed entirely from their algebra of functions $\Gamma(X)$.  This is the case, for example, when $X$ has global coordinate functions.  Conversely, \emph{every} commutative algebra $A \in \CAlg$ can be thought of as coordinates on a space, which we call $\Spec(A)$.  Together, this assembles into an adjunction
\begin{equation}\label{eq:adjunction}
     \begin{tikzcd}
        \CAlg^{\mathrm{op}}\ar[r,shift left=2,"{\Spec}"{name=A}]& \Sch\,\ar[l,shift left=2,"{\Gamma}"{name=B}].\ar[phantom, from=A,to=B,"\dashv"rotate=90]
    \end{tikzcd}
\end{equation}
It turns out that the functor $\Spec$ is fully faithful, so the above adjunction gives an equivalence between $\CAlg^{\mathrm{op}}$ and a full subcategory $\Aff\subset\Sch$ on those spaces which are ``globally trivial''.  These spaces are called \emph{affine}.  

\begin{example}
The Euclidian spaces $\mathbb{R}^n$, considered as real manifolds, are affine.  
\end{example}

Every $X\in \Sch$ is locally affine, so it can be reassembled by gluing together affines.  This allows us to think of affines as ``coordinate patches''.  There is a natural categorical way of encoding gluing.  For every $X\in \Sch$, there is a functor  
\begin{equation}
   \Aff^{\mathrm{op}}\subset \Sch^{\mathrm{op}}\xrightarrow{\Hom_{\Sch}(-,X)} \Set, 
\end{equation}
which should be thought of as mapping every affine $S\in \Aff$ to the list of all copies of $S$ used to assemble $X$.  Doing this for everything in $\Sch$, we obtain a functor 
\begin{align}
    \Sch\xrightarrow{\yo} \Fun(\Sch^{\mathrm{op}},\Set)\xrightarrow{\mathrm{res}}\Fun(\Aff^{\mathrm{op}},\Set)&&X\mapsto \Hom_{\Sch}(-,X). 
\end{align}
By Yoneda lemma, the assignment $\yo$ is fully faithful.  In fact, one did not need to remember the value on every space: it was enough to remember the value on affines, since $\mathrm{res}\circ \yo$ also turns out to be fully faithful.  In this sense \emph{every} $F\in \Fun(\Aff^{\mathrm{op}},\Set)$ can be thought of as a rule of how to glue affines.  Such a functor is called a \emph{presheaf}.  

A shortcoming of presheaves is that they do not encode quotient spaces with stabilizers\footnote{This is true in differential geometry also: the quotient of a manifold is not generally a manifold, but an \emph{orbifold}.}: the prototypical example is the classifying stack $\rB G = \pt /G$.  To combat this deficiency, we consider presheaves which remember their points, but also \emph{automorphisms} of each point.  This is achieved by replacing \emph{sets} with \emph{groupoids}, which are exactly ``sets with automorphisms''.  A functor $F\in \Fun(\Aff^{\op},\Spaces)=:\prstk$ is called a \emph{prestack}, and can be thought of as assigning to each $S\in \Aff$ the list of copies of $S$ used to assemble a geometric space, \emph{together with their automorphisms}. In particular  there is a chain of fully faithful inclusions 
\begin{equation}
    \Aff\to \Sch \xrightarrow{\yo}\prstk. 
\end{equation}

\subsection{Chain Complexes, Stable Categories, and Spectra}
By replacing sets with $\infty$-groupoids, we have given the appropriate homotopy-coherent notion of \emph{target} for prestacks.  To arrive at derived algebraic geometry, we must also replace $\CAlg$ with an $\infty$-analog.  We give an overview of the $\infty$-analog of \emph{abelian categories}, the category of \emph{abelian groups} $\Ab$, and the ring of integers $\mathbb{Z}$.  The these objects are replaced by \emph{stable $\infty$-categories}, the category of \emph{spectra} $\Sp$, and the \emph{sphere spectrum} $\mathbb{S}$, respectively. 

\begin{table}[H]
    \centering
\begin{tabular}{c||c|c|c}
         Classical & Abelian categories& $\Ab$& $\mathbb{Z}$\\\hline
         Derived & Stable $\infty$-categories& $\Sp$& $\mathbb{S}$ 
    \end{tabular}
    \caption{Classical objects and their (derived) $\infty$-analog. }
    \label{tab:placeholder}
\end{table}
Spectra provide a homotopical upgrade of chain complexes.  Classical objects are recovered by passing to the degree-$0$ component.  In particular, $\pi_0(\mathbb{S})=\mathbb{Z}$.  A spectrum is \emph{connective} if its homotopy groups in negative degree are trivial.  Denote by $\Sp^{\cn}\subset \Sp$ the full subcategory on the connective spectra. 

Just as an ordinary commutative algebra is an algebra object in abelian groups, a \emph{derived} commutative algebra is a commutative algebra object in spectra, and we can consider them as the ``coordinate patches'' from which we assemble geometry. 
\begin{definition}
The category of \emph{connective $E_\infty$-ring spectra} is $\CAlg(\Sp^{\cn})$. 
\end{definition}
\begin{definition}
The category of \emph{derived prestacks} is $\Fun(\CAlg(\Sp^{\cn})^{\mathrm{op}},\Spaces)$. 
\end{definition}
In the remainder of the paper, $\CAlg$ and $\prstk$ will refer to the categories of $E_\infty$-rings and derived prestacks, respectively.  The categories of classical commutative rings and classical schemes embed fully faithfully into their derived analogs, and so a reader uninterested in derived behaviour can just as easily pretend that $\CAlg$ \emph{is} its classical analog.

\begin{definition}
 For $\infty$-category $\cC$ with terminal object $\pt\in \cC$, and morphisms $f:X\to Y$ and $\pt\to Y$ in $\cC$, a fiber (resp. cofiber) is a homotopy pullback (resp. pushout) diagram 
 
      $$\begin{tikzcd}
       \fib(f)\ar[dr,phantom,very near start,"\lrcorner"]\ar[r]\ar[d]&X\ar[d,"f"]\\
       \pt\ar[r]&Y
    \end{tikzcd}
\hspace{3cm}
    \begin{tikzcd}
        X\ar[r,"f"]\ar[d]&Y\ar[d]\\
        \pt\ar[r]&\cofib(f)\,.\ar[ul,phantom, very near start, "\ulcorner"]
    \end{tikzcd}
  $$ 
\end{definition}

An \emph{abelian} category is an ordinary category which has sums, kernels, and cokernels which behave nicely.  The following is \cite[Definition 1.1.1.9]{LurieHA}. 
\begin{definition}
Let $\cC$ be an $\infty$-category.  $0\in \cC$ is called a \emph{zero object} if it is both initial and terminal.  $\cC$ is called \emph{stable} if it has a zero object, every morphism in $\cC$ has a (co)fiber, and a square is a fiber square if and only if it is a cofiber square. 
\end{definition}

\begin{example}
Let $R \in \CAlg$ be a commutative ring.  Then the $\infty$-category of chain complexes of $R$-modules is stable, and agrees with $\Mod_R(\Sp)$ by \cite[Theorem 7.1.2.13]{LurieHA}.  In particular, for a field $\mathbb{K}$, the derived category of $\mathbb{K}$-vector spaces is stable.
\end{example}

The following is defined in \cite[Section 5.5]{LurieHTT}.  An alternate characterization is given by \cite[Proposition 12.1.4]{stefanich:thesis}. 
\begin{definition}
A category is called \emph{presentable} if it has all small colimits and is generated by a small collection of objects, in a suitable sense\footnote{By this we mean that it is the completion of a $\kappa$-small category under $\kappa$-filtered colimits, for some regular cardinal $\kappa$.  See \cite[Proposition 5.4.2.2]{LurieHTT}.  This is a set-theoretic condition which ensures that the adjoint functor theorem can be applied. \cite[Corollary 5.5.2.9]{LurieHTT}.}.  Denote by $\Pr^\mathrm{L}$ the category whose objects are presentable categories, and morphisms are colimit-preserving functors.  Denote by $\Pr^{\mathrm{L}}_{\st}\subset \Pr^\mathrm{L}$ the full subcategory on those presentable categories which are also stable. 
\end{definition}

\begin{example}
    If $\cC$ is a small $\infty$-category, then $\Fun(\cC^{\op},\Spaces)$ is presentable. 
\end{example}

The category $\Pr^\mathrm{L}_{\st}$ has a natural symmetric monoidal structure, with monoidal unit given by $\Sp$.  Due to Stefanich, one can inductively define an $n$-dimensional version, the category of \emph{presentable stable $n$-categories} $\nPr^\mathrm{L}_\st$, whose monoidal unit is $\mathbf{(n-1)}\Pr^\mathrm{L}_\st$.  The definition of $\nPr^\mathrm{L}_\st$ is quite subtle, and we will not address here, instead referring to \cite[Chapter 12]{stefanich:thesis} for details. 

%%%%%%%%%%%%%%%%%%%%%%%%%%%%%%%%%%%
\subsection{Quasi-coherent Sheaves}\label{subsection:QCoh}
%%%%%%%%%%%%%%%%%%%%%%%%%%%%%%%%%%%%
Given a vector bundle $V$ on $X\in \Sch$, the space of \emph{global sections of $V$}, given by $\Gamma(X,V)$, is naturally a module over the algebra of functions $\Gamma(X)$.  In fact, the \emph{Serre-Swan theorem} states that there is an \emph{equivalence} between finite-rank vector bundles on $X$ and finitely generated projective $\Gamma(X)$-modules.  This is naturally generalized by allowing non-projective modules, and allowing the module to vary between coordinate patches.  This results in the notion of \emph{quasi-coherent sheaf}, which can be supported on subspaces of $X$.  We begin by defining quasi-coherent sheaves on affine schemes, and then globalize by gluing. 
\begin{definition}
For an affine scheme $\Spec(A)$, the \emph{category of quasi-coherent sheaves} on $\Spec(A)$ is defined to be $\Qcoh(\Spec(A)):=\Mod_A$, the category of $A$ modules in $\Sp^{\cn}$.  
\end{definition}
For any $X\in \prstk$, quasi-coherent sheaves over $X$ should be determined by consistently gluing sheaves on the coordinate patches of $X$. 
\begin{equation}
   \Qcoh(X)\simeq \lim_{S\in \Aff^{\op}_{/X}}\Qcoh(S)\,.
\end{equation}
This gives a way of extending the functor $\Mod:\Aff^{\op}\to \Pr^\rL_\st$ to all of $\prstk$.  Extending a functor from a subcategory by gluing in this way is called \emph{right Kan extension}.  Then the category of quasi-coherent sheaves is \emph{defined} by this extension property.  This definition of quasi-coherent sheaves has a higher-dimensional analog, which we will simultaneously give here.  We refer to \cite{stefanich:thesis,StefanichSheaf} for further details in the case $n>1$. 
\begin{definition}
   Define the functor $\Mod^n:\CAlg(\Sp)\to \nPr^\mathrm{L}_{\st}$ taking a commutative algebra to its $n$-iterated category of modules\footnote{$\Mod^n$ can also be defined by delooping and cocompleting.}.  Define the functor $\nQcoh:\prstk^{\op}\to \nPr^\mathrm{L}_{\st}$ to be the right Kan extension\footnote{This extension exists by \cite[Proposition 14.2.2]{stefanich:thesis}.} of $\Mod^n$ given in the following diagram 
\begin{equation}
\begin{tikzcd}
     \prstk^{\op} \arrow[rd,"\nQcoh",dotted]& \\
    \CAlg \arrow[r,"\Mod^n",swap] \arrow[u,"\yo"]& \nPr_\st^{\mathrm{L}}\,.
\end{tikzcd}
\end{equation}
The functor $\nQcoh$ takes a stack $X$ to the category of \emph{quasi-coherent sheaves of $(n-1)$-categories on $X$. }This upgrades to a functor $\prstk^{\op}\to \CAlg(\nPr^\rL_\st)$: for each $X\in \prstk$, $\nQcoh(X)$ inherits a symmetric monoidal structure, which we shall denote by $\nQcoh^\otimes(X)$.   
\end{definition}

\begin{example}\label{example:RepQcohBG}
   Let $G$ be an algebraic group, and consider the associated classifying stack $\rB G$.  There is an equivalence of symmetric monoidal categories 
   \begin{equation}
      \Rep(G)\simeq \Qcoh^\otimes(\rB G).  
   \end{equation}
\end{example}
\noindent For a morphism $f:X\to Y$ in $\prstk$, functoriality of $\nQcoh$ gives an adjunction
\begin{equation}
    \begin{tikzcd}
        \nQcoh^\otimes(Y)\ar[r,shift left=2,"f^*"{name=A}]&\nQcoh^\otimes(X)\ar[l,shift left=2,"f_*"{name=B}]\ar[phantom,from=A, to=B, "\dashv"rotate=-90], 
    \end{tikzcd}
\end{equation}
Where the morphisms $f^*$ and $f_*$ are called \emph{pullback} and \emph{pushforward}, respectively.  In the case of the canonical map $X\to \Spec(\mathbb{S})$, this induces a morphism 
\begin{equation}
    \Gamma(X,-):\nQcoh^\otimes(X)\to \nQcoh^\otimes(\Spec(\mathbb{S}))\simeq \nPr^\mathrm{L}_{\st}
\end{equation}
which we call the \emph{global sections} functor.  Just as affine schemes were completely recovered from their global functions, one can ask when a stack can be recovered from global sections of its ``categorified functions''.  The following definition is due to \cite{Gaitsgory_2015} in the case $n=1$ and \cite{stefanich:thesis} for arbitrary $n$. 

\begin{definition}
   A stack $X$ is called \emph{n-affine} if the global sections functor $\Gamma(X,-):\nQcoh^\otimes(X)\to \nPr^\mathrm{L}_\st$
   is monadic in $\mathbf{(n+1)}\Pr^\mathrm{L}_\st$.  
\end{definition}
This is to say, $X$ is $n$-affine if and only if $\nQcoh^\otimes(X)\simeq \Mod_{\noQcoh^\otimes(X)} (\nPr^\mathrm{L}_\st)$, where we think of $\noQcoh^\otimes(X)$ as a categorified algebra.

\begin{example}
    Suppose that $G$ is an affine algebraic group over characteristic $0$ field.  Then $\rB G$ is $1$-affine by \cite[Theorem 2.2.2]{Gaitsgory_2015}.  
\end{example}

\subsection{Descent}\label{subsection:descent}
Not every prestack can be thought of as geometric.  Geometric spaces satisfy an additional locality property: they can be reconstructed from an open cover.  For a functor $F$ out of $\Aff^{\op}$ and a map of affines $f:X\to Y$, we would like to understand when $F(Y)$ can be reconstructed from the potential cover $f$.  This is called \emph{descent}. 

Let $\Fin$ be the category of non-empty finite sets, and $\Delta$ be the category of non-empty finite \emph{ordered} sets and order preserving maps.  There is a canonical functor $\Delta\to \Fin$.  Let $\cC$ be a category with limits of finite sets, consider $c\in \cC$.  For any finite set $I\in \Fin$, we can define the \emph{cotensor} $c^I$ which is the limit of the constant diagram valued at $c$.  This gives a functor $\Delta^{\op}\to \Fin^{\op}\to \cC$.  

\begin{definition}
Let $\cD$ be a category with pullbacks. For $Y\in \cD$ and $f:X\to Y$ in $\cD_{/Y}$, applying the above construction gives a functor $\Delta^{\op}\to \cD_{/Y}\to \cD$, called the \emph{\v{C}ech nerve} of $f$, and denoted by $X^\bullet/Y$. 
\end{definition}

\begin{definition}
We say that a functor $F:\prstk^{\op}\to \cD$ \emph{satisfies descent} along $f:X\to Y$ (or alternatively, that $f$ satisfied descent along $F$), if the canonical limit map is an equivalence
\begin{equation}
F(Y)\xrightarrow{\sim}\lim_{\Delta^{\op}}F(X^\bullet/Y). 
\end{equation}
\end{definition}

We do not expect to be able to reconstruct a geometric space from any map $f:X\to Y$, but only the \emph{covering maps}.  Then we must make a choice of what it means for one affine to cover another.  Reasonable choices for notions of covering maps are \emph{fppf, \'etale, smooth} or \emph{Zariski} morphisms.  We will refer to \cite[section 2.1]{GaitsgoryRozenblyum} for definitions. 

\begin{definition}
    A prestack $F\in \prstk= \Fun(\Aff^{\op},\Spaces)$ is called a (fppf/\'etale/Zariski) \emph{stack} if it takes colimits of finite sets to limits, and $F$ satisfies descent along all (fppf/\'etale/Zariski) maps in $\Aff$.  Denote by $\Stk\subset \prstk$ the full subcategory on the stacks.
\end{definition}
Stack should be thought of as gluings of coordinate patches that are local with respect to covers, and so are \emph{generalized geometric spaces}.  There is a chain of fully faithful inclusions\footnote{$\Aff$ includes into $\Stk$ only when the chosen notion of topology is \emph{subcanonical}.  This will be true in our case. }
\begin{equation}\label{eq:glueaffine}
    \Aff\subset \Sch\subset \Stk\subset \prstk,  
\end{equation}
which tells us that the geometric spaces we began with still behave the same when considered as stacks.

\begin{warning}
The chosen notion of cover (fppf/ \'etale/ Zariski) \textbf{does} result in a different category $\Stk$, and so, in principle, our notation should reflect this.  We choose to leave this choice implicit, only referring to the chosen notion of cover when necessary. 
\end{warning}

\subsection{Convolution}\label{subsection:Convolution}

Consider two stacks equipped with a map $p:X\to Y$ satisfying descent with respect to $\nQcoh$, and consider the following diagram with projection maps:
\begin{equation}
\begin{tikzcd}
&X\times_Y X\times_Y X\ar[dl,"p_{12}"]\ar[d,"p_{13}"]\ar[dr,"p_{23}"]&\\
X\times_Y X&X\times_Y X&X\times_Y X\,.
\end{tikzcd}    
\end{equation}
One obtains a \emph{convolution} monoidal structure on $\nQcoh(X\times_Y X)$, defined by pulling back sheaves along the maps $p_{12}$ and $p_{23}$, and then pushing forward along remaining factor, to give 
\begin{equation}
    M\star N:= {p_{13}}_*(p_{12}^*(M)\otimes p_{23}^*(N)). 
\end{equation}
Denote by $\nQcoh^\star(X\times_Y X)$ the category of quasi-coherent sheaves with this monoidal structure. 

\begin{example}
     Let $G$ be an algebraic group.  There is an isomorphism $G\simeq \pt\times_{\rB G}\pt$.  This induces a convolution monoidal structure 
    \begin{equation}
        \nQcoh^\star(G):=\nQcoh^\star(\pt\times _{\rB G}\pt). 
    \end{equation}
    In the case of finite group $G$, this can be identified with the usual monoidal structure on the category of $G$-graded $n$-vector spaces $\nVect_G$. 
\end{example}

\begin{example}
For a sheaf $X$, the \emph{free loop space} is defined to be 
\begin{equation}
    \mathcal{L}X:=X\times_{X\times X}X\simeq \Hom(\rB \mathbb{Z},X). 
\end{equation}
  There is a convolution monoidal structure on the category of higher quasi-coherent sheaves over $\cL X$ which we denote by
  \begin{equation}
      \nQcoh^\star(\mathcal{L} X):=\nQcoh^\star(X\times_{X\times X}X). 
\end{equation}
\end{example}

\subsection{Correspondences and integral transforms}\label{sec:integraltransforms}
We now introduce the correspondence categories, which can be used to encode the additional functoriality of sheaf theories.  See \cite[Section 10.1]{stefanich:thesis} for further details. 
\begin{definition}
For $\cC$ a category with pullbacks, $\Corr(\cC)$ is the category with objects coinciding with $\cC$, and morphisms between $U,W$ in $\cC$ given by \emph{correspondences}: diagrams in $\cC$ of the form 
\begin{equation}\label{eq:spanstack}
    \begin{tikzcd}[row sep=0.5]
        & V \arrow[rd,"f"] \arrow[ld,"g",swap]& \\
        U & & W\,.
    \end{tikzcd}
\end{equation}
Composition is given by pullback.  $\Corr(\cC)$ is equipped with inclusions $\iota_\cC: \cC \to \Corr(\cC)$ and $\iota^R_\cC: \cC^{\op} \to \Corr(\cC)$.  When $\cC$ is monoidal, $\Corr(\cC)$ also inherets a monoidal structure.  
\end{definition}
\begin{definition}
The higher correspondence category $\nCorr(\cC)$ is defined inductively for $n> 2$: objects coincide with those of $\cC$.  For objects $U,W$ in $\cC$, the morphism $(n-1)$-category is 
\begin{equation}
\Hom_{\nCorr}(U,W):=\mathbf{(n-1)Corr}(\Hom_{\Corr(\cC)}(U,W))
\end{equation}
Furthermore, for each $m\leq n$ there are natural transformations 
\begin{align}
\iota^{m,n}_\cC: \mathbf{mCorr}(\cC) \to  \mathbf{nCorr}(\cC)&&(\iota^{m,n}_\cC)^R:  \mathbf{mCorr}(\cC)^{\op} \to  \mathbf{nCorr}(\cC).
\end{align}
See \cite[Section 11.1]{stefanich:thesis} for details.  
\end{definition}
\begin{definition}
    A morphism of stacks $f:X\to Y$ is called \emph{affine schematic} if for every affine scheme $S$ with a map $S\to Y$, the pullback $S\times_Y X$ is representable by an affine scheme.  Denote by $\prstk_{\mathrm{rep}}\subset \prstk$ the wide subcategory on the affine-schematic morphisms.
\end{definition}
For $n\geq 2$, $\nQcoh$ can be extended to a functor with domain $\mathbf{(n+1)Corr}(\prstk)$ \cite[Theorem 14.2.9]{stefanich:thesis}.  This, together with the dualizability properties of $\nCorr$ \cite{HaugsengIteratedSpans,AKSZ2022}, allow us to construct extended TQFTs.  

Consider a correspondence of prestacks as in \Cref{eq:spanstack}.  An object $K \in \nQcoh^\otimes(V)$ can be used to  define a ``pull-push'' morphism $\nQcoh^\otimes(U)\to \nQcoh^\otimes(W)$ by the formula $F\mapsto g_*(f^*(F) \otimes K)$.  In context of sheaves of functions, $f^*$ is given by integration, and so the transform is implemented by convolution against the \emph{integral kernel} $K$\footnote{The classical \emph{Fourier transform}, and more generally, Fourier-Mukai transform \cite{Mukai1981}, can be realized this way. }.  For this reason, these functors are called \emph{integral transforms}.  A theorem of \cite{ben2010integral} states that every functor between categories of sheaves on sufficiently nice spaces arises this way.  We will require the following generalization. 

\begin{proposition}[{\cite[Proposition 4.8]{StefanichSheaf}}]\label{Prop:BFN}
    Suppose that $X,Y,Z$ are prestacks and consider $X\to Y\leftarrow Z$.  Suppose that Y is $n$-affine.  Then
    \begin{equation}
\nQcoh^\otimes(X)\otimes_{\nQcoh^\otimes(Y)}\nQcoh^\otimes(Z)\simeq \nQcoh^\star(X\times_Y Z). 
    \end{equation}
\end{proposition}

\begin{corollary}\label{cor:nQCohDual}
Suppose that $Y$ is n-affine, and consider map of stacks $X\to Y$.  Then $\nQcoh^\otimes(X)$ is canonically self dual as a $\nQcoh^\otimes(Y)$-module.  
\end{corollary}
\begin{proof}
The result follows from Proposition \ref{Prop:BFN} using the proof of \cite[Corollary 4.8]{ben2010integral} (or alternatively \cite[Proposition 10.3.3]{stefanich:thesis}). 
\end{proof}

\begin{corollary}\label{corr:IntegralFunctor}
    Suppose that $X,Y,Z$ are stacks and consider morphisms $X\to Y\leftarrow Z$.  Suppose also that Y is $n$-affine, then
    \begin{equation}
\Fun_{\nQcoh^\otimes(Y)}(\nQcoh^\otimes(X),\nQcoh^\otimes(Z))\simeq \nQcoh^\star(X\times_Y Z). 
    \end{equation}
\end{corollary}

\begin{proof}
    It follows from self-duality of $\nQcoh^\otimes(X)$ that we have an equivalence of functors 
    \begin{align}
        \Fun_{\nQcoh^\otimes(Y)}(\nQcoh^\otimes(X),-)&\simeq \nQcoh^\otimes(X)^\vee\otimes_{\nQcoh^\otimes(Y)}(-)\\ \notag
        &\simeq \nQcoh^\otimes(X)\otimes_{\nQcoh^\otimes(Y)}(-)\,.
    \end{align}
\end{proof}

%%%%%%%%%%%%%%%%%%%%%%%%%%%%%%%%%%
\section{$(n+1)$-Dimensional Gauging and  Symmetries}\label{section:dynamicalgauging}
%%%%%%%%%%%%%%%%%%%%%%%%%%%%%%%%%%

%%%%%%%%%%%%%%%%%%%%%%%%%%%%%%%%%%
\subsection{Constructing the TQFT in Arbitrary Dimensions}\label{subsection:generalTQFT}
%%%%%%%%%%%%%%%%%%%%%%%%%%%%%%%%%%

For a theory in $(n+1)$d with a $G$-symmetry   we now construct a once categorified $(n+1)$d TQFT, which we view as a partially defined-$(n+2)$d TQFT.  Apriori, this $(n+2)$d TQFT can be evaluated on thickenings of a $(n+1)$-manifold by an interval.  This is sufficient to allow an interpretation as the $(n+2)$d {SymTFT} for an $(n+1)$d boundary theory.  
%\matt{We should sign post some notation: 1) what the subscript on$\cZ$ means. 2) what the superscript on $\Pr$ means. 3) how we are moving away from the morita category}

\begin{warning}\label{warning:shiftindex}
In what follows, the once-categorified $(n+1)$d TQFT will be constructed from $\npRep(G)$ rather than $\nRep(G)$, as one would have expected from reading \S\ref{subsection:buldTFT}.  The reason for this index shift is because the ``Morita TQFTs'' are constructed from $\nRep(G)$, but for technical reasons, the TQFTs in \Cref{thm:generalTQFT} will be valued in $\mathbf{(n+1)Pr}^\rL_{\st}$, and the functor $\Mod:\Mor(\nPr^\mathrm{L}_\st)\to \mathbf{(n+1)}\Pr^\mathrm{L}_\st$ induces an index shift, since $\Mod_{\nRep(G)}\simeq \mathbf{(n+1)}\Rep(G)$.  

We show in  \Cref{prop:MoritaGeneralTFT} that in the case of  $n=1$ the TQFT constructed from \Cref{thm:generalTQFT} agrees with the one constructed from \Cref{prop:buildTQFT}.  The two formalisms also agree for general $n$.
\end{warning}

\begin{theorem}\label{thm:generalTQFT}
    Let $G$ be an affine algebraic group over a characteristic $0$ field $\mathbb{K}$.  For all $n\geq 1$, there is a fully extended once-categorified $(n+1)$-dimensional TQFT 
   \begin{equation}
       \cZ_{n+2}^G:\Bord^{\mathrm{fr}}_{n+1}\to \mathbf{(n+1)}{\Pr^\mathrm{L}_\st}^{(n+2)\op}. 
   \end{equation} 
  For $0\leq i\leq n$, the value of $\cZ_{n+2}$ on a closed $i$-manifold $M^{i}$ is given by 
    \begin{equation}
        \cZ^G_{n+2}(M^{i})\simeq \mathbf{(n-i+1)}\Qcoh^\otimes(\mathrm{Maps}(M^{i},\rB G)), 
    \end{equation}
higher quasi-coherent sheaves on the moduli space of flat $G$ connections on $M^{i}$.  In particular, 
\begin{equation}
    \cZ_{n+2}^G(\pt)=\mathbf{(n+1)}\Rep(G). 
\end{equation}
\end{theorem}

In this theorem, the lower $(n+2)$ index in $\cZ^G_{n+2}$ denotes the ``physical'' dimension of the TFT.  The upper index $(n+2)\mathrm{op}$ on $\mathbf{(n+1)}{\Pr^\mathrm{L}_\st}^{(n+2)\op}$ indicates that the direction of $(n+2)$-cells are reversed from $\mathbf{(n+1)}\Pr^\mathrm{L}_\st$.

\begin{proof}
Let $X$ be a prestack.  By \cite[Corollary 4.4.1]{AKSZ2022}, there is a functor $\Bord_{n+1}^{\mathrm{fr}}\xrightarrow{\mathrm{Maps}(-,X)} \mathbf{(n+1)Corr}(\prstk)$. 
%\dev{confirm this claim}
If $X$ has affine diagonal, this functor factors through $\mathbf{(n+1)Corr}(\prstk_{\mathrm{rep}})$.  Then for any stack $X$ with affine diagonal, we obtain a fully extended once-categorified $(n+1)$-dimensional TQFT given by the composition
%\footnote{We elaborate on the meaning of the correspondence category $\mathbf{(n+1)Corr}$ in  \S\ref{sec:integraltransforms}.} 
$$\Bord_{n+1}^{\mathrm{fr}}\xrightarrow{\mathrm{Maps}(-,X)} \mathbf{(n+2)Corr}(\prstk_{\mathrm{rep}})\xrightarrow{\npQcoh} \mathbf{(n+1)}{\Pr^\mathrm{L}_\st}^{(n+2)\op},$$
 where the right map is the unique extension of $\nQcoh$ given by \cite[Theorem 14.2.9]{stefanich:thesis}.  The value on closed $i$-manifolds follows from the definition of the above composition and $\nQcoh$ (see \cite{StefanichSheaf}).  In particular, the stack $\rB G$ determines a theory in this way, which we denote by $\cZ^{G}_{n+2}$. 
\end{proof}

\begin{notation}
    The object $\Spec(\mathbb{K})\in \prstk$ determines a fully-extended once-categorified $(n+1)$-dimensional TQFT as in \Cref{thm:generalTQFT}, which we will call the \emph{trivial theory}, and denote by $\cZ_{n+2}^\mathrm{triv}$.  A \emph{boundary} is a morphism of theories with target $\cZ^\mathrm{triv}_{n+2}$. 
\end{notation}

\begin{lemma}\label{lemma:inductivedescent}
   Suppose that $f:X\to Y$ is a map of prestacks satisfying descent with respect to $\nQcoh$ for $n\geq 2$.  Then $f^*:\npQcoh^\otimes(Y)\to \npQcoh^\otimes(X)$ is monadic.  
\end{lemma}

The following is adapted from the proof of \cite[Proposition 14.3.3]{stefanich:thesis}.
\begin{proof}
   We will confirm the conditions of enriched Barr-Beck-Lurie \cite[Theorem 7.4.10]{stefanich:thesis}.  By \cite[Corollary 14.2.10]{stefanich:thesis}, $f^*$ admits a left adjoint (and preserves colimits since it is left adjoint).  Then it is sufficient to check that $f^*$ is conservative.  To do so it suffices to check that $f_*$ generates $\npQcoh^\otimes(Y)$ under colimits and tensors.  Since $f_*$ is a map of $\npQcoh^\otimes(Y)$ modules, it suffices to confirm that the monoidal unit $\nQcoh^\otimes(Y)$ is a colimit of objects in the image of $f_*$.  Since $f$ satisfies descent, the augmented simplicial object $\lim \nQcoh^\otimes(X^\bullet/Y)$ is a limit diagram.  Applying \cite[Proposition 10.3.1, Theorem 14.2.9]{stefanich:thesis}, we see that this diagram is adjointable.  By \cite[Theorem 8.5.3]{stefanich:thesis}, $\npQcoh^\otimes(Y)$ satisfies the passage to adjoints property, so $\nQcoh^\otimes(Y)$ is the colimit of the adjoint of the diagram $\nQcoh^\otimes(X^\bullet/Y)$.  Then $\nQcoh^\otimes(Y)$ is the colimit of objects in the image of $f_*$, so $f^*$ is conservative, and the result follows. 
\end{proof}

\begin{corollary}\label{cor:BGnAff}
    Let $G$ be an affine algebraic group over characteristic $0$ field $\mathbb{K}$.  Then $\rB G$ is $n$-affine for all $n\geq 1$.
\end{corollary}

\begin{proof}
    The case $n=1$ follows from \cite[Proposition 2.2.2]{Gaitsgory_2015}.  Let $n\geq 2$.  Since $\Spec(\mathbb{K})$ is $n$-affine and composition of monadic morphisms is monadic, it suffices to confirm that $\nQcoh^\otimes(\rB G)\to \nQcoh(\Spec(\mathbb{K}))$ is monadic. We will show this via the methods in \Cref{lemma:inductivedescent}, which means we would like to show that $\Spec(\mathbb{K})\to \rB G$ satisfies $(n-1)$-descent. To show this, it suffices to confirm that $\Spec(\mathbb{K})\to \rB G$ is fppf.  To show fppf, it suffices to confirm that the base change along a cover of $\rB G$ is faithfully flat.  Then $\Spec(\mathbb{K})\to \rB G$ is fppf since $G\to \Spec(\mathbb{K})$ is.  The result follows from \Cref{lemma:inductivedescent}, since $\nQcoh$ satisfies fppf descent.  
\end{proof}

\begin{proposition}\label{prop:MoritaGeneralTFT}

In the case $n=1$, the theory constructed in \Cref{thm:generalTQFT} can be identified with the TQFT constructed in \Cref{prop:buildTQFT} under the inclusion $\Mor(\Pr^{\mathrm{L}}_\st)\to \mathbf{2Pr}^\mathrm{L}_\st$. 
\end{proposition}

\begin{proof}
In order to make the identification with the TQFT constructed in \Cref{prop:buildTQFT} we note that
there is a fully faithful functor $\Mod:\Mor(\nPr^\mathrm{L}_\st)\to \mathbf{(n+1)}\Pr^\mathrm{L}_\st$.  Under this inclusion, the image of $\nRep(G)$ is given by $\Mod_{ \nRep(G)}\simeq \mathbf{(n+1)}\Rep(G)$, using that $\rB G$ is $n$-affine (\Cref{cor:BGnAff}).  By cobordism hypothesis, the composition of the TQFT constructed in \Cref{prop:buildTQFT} with $\Mod$ agrees with the one constructed in \Cref{thm:generalTQFT}. 
\end{proof}

\begin{lemma}\label{lemma:ConstructWall}
 Any correspondence in $\prstk$
   \begin{equation}
       \begin{tikzcd}[row sep=0.5]
           &X\ar[dl,swap,"\alpha"]\ar[dr,"\beta"]&\\
           Y&&Z
       \end{tikzcd}
   \end{equation}
   determines a twice-categorified domain wall between TQFTs constructed from $Y$ and $Z$.  If $\alpha,\beta$ are affine schematic, then the resulting domain wall can be extended to a once-categorified domain wall. 
\end{lemma}

\begin{proof}
   For all $n\geq 1$, the above correspondence determines a fully-adjunctible $1$-morphism in $\mathbf{(n+1)Corr}(\prstk)$ by \cite[Proposition 11.1.9]{stefanich:thesis} (or alternatively \cite{HaugsengIteratedSpans}).  Applying the functor  $\npQcoh:\mathbf{(n+1)Corr}\to \mathbf{(n+1)}\Pr^\mathrm{L}_\st$ of \cite[Theorem 14.2.9 (i)]{stefanich:thesis}, we obtain a fully-adjunctible $1$-morphism in ${\mathbf{(n+1)}\Pr^\mathrm{L}_\st}^{\leq n-1}$. 
   
   If $\alpha,\beta$ are affine schematic, then the above correspondence determines a fully adjunctivle morphism in $\mathbf{(n+2)Corr}(\prstk_\mathrm{rep})$.  Applying the functor $\npQcoh:\mathbf{(n+2)Corr}\to {\mathbf{(n+1)}\Pr^\mathrm{L}_\st}$ of \cite[Theorem 14.2.9 (ii)]{stefanich:thesis}, we obtain a fully-adjunctible $1$-morphism in ${\mathbf{(n+1)}\Pr^\mathrm{L}_\st}^{\leq n}$.  

   In both cases, the result follows from cobordism hypothesis with singularities \cite[Section 4.3]{LurieTFT}. 
\end{proof}

\begin{corollary}
   Let $G$ be an affine algebraic group over characteristic $0$ field.  For every $n\geq 1$, any closed subgroup $H\subset G$ determines a twice-categorified boundary for $\cZ^G_{n+2}$.  Endomorphisms of this boundary coincide with $\End^\mathrm{L}_{\nRep(G)}(\nRep(H))$. 
\end{corollary}

\begin{proof}
By \Cref{lemma:ConstructWall}, the following correspondence determines a twice-categorified domain wall between $\cZ^G_{n+2}$ and $\cZ^\mathrm{triv}_{n+2}$.
    \begin{equation}
       \begin{tikzcd}[row sep=0.5]
           &\rB H\ar[dl]\ar[dr]&\\
           \rB G&&\Spec(\mathbb{K}). 
       \end{tikzcd}
   \end{equation}   
   Under the functor $\Mod:\Mor(\nPr^\mathrm{L}_\st)\to \mathbf{(n+1)}\Pr^\mathrm{L}_\st$, the above morphism in $\mathbf{(n+1)}\Pr^\mathrm{L}_\st$ is identified with the $\nRep(G)$-module category $\nRep(H)$.  Since the functor $\Mod$ is fully faithful, endomorphisms of the above boundary can equivalently be calculated in $\Mod_{\nRep(G)}$ as endomorphisms of $\nRep(H)$. 
\end{proof}

\begin{rem}
    As in the case $n=1$, it may be possible to show that many of the $\nRep(H)$ extend to give once-categorified boundaries, instead of twice-categorified.  Relevant dualizability results appear in \cite{StefanichDualizabilityStacks, Hall_2015}. 
\end{rem}

%%%%%%%%%%%%%%%%%%%%%%%%%%%%%%%%%%%%%%%%%%%
\subsection{Electric symmetries}\label{subsection:extractingelectric}
%%%%%%%%%%%%%%%%%%%%%%%%%%%%%%%%%%%%%%%%%%%

In this section we prove \Cref{mainthm:Wilson} and \Cref{prop:linkingelectric}, and we expand upon the physical relevance in terms recovering the 1-form center symmetry. We end the first part of this section by identifying the $(n-1)$-form electric symmetry of $(n+1)$d flat $G$-gauge theory, and in the second part we identify the 1-form center symmetry.

Before obtaining these results we first expound on the theory we develop in \S\ref{subsection:electricsym}  and \S\ref{subsection:morita} in the case where groups are involved. The reader is directed to those sections for a more general treatment of what is in this section.
In the following, we will assume that $G$ is an affine algebraic group over characteristic $0$ field and $H,H'\subset G$ are closed subgroups, unless explicitly stated otherwise.  

\iffalse
In this case, $\rB G,\rB H$ satisfy the conditions of the following premise. 
\begin{assumption}\label{ass:qcadft}
In the following, we will consider quasi-compact stacks with affine diagonal of almost finite presentation over characteristic 0 field $\mathbb{K}$.  $\nQcoh$ defines a sheaf theory on this subcategory of $\Stk$, as in \cite[Example 4.6]{StefanichSheaf}.  
\end{assumption}
\fi

\begin{notation}
     Let $G$ be an algebraic group and $H,H'\subset G$ subgroups.  Define the (stacky) \emph{double coset} to be $H\backslash G/H':= \rB H\times_{\rB G}\rB H'$.  In the following, the symbol $H\backslash G/H$ will refer to the stack quotient in this sense.  This induces a convolution monoidal structure, which we denote by 
    \begin{equation}
        \nQcoh^\star(H\backslash G/H):= \nQcoh^\star(\rB H\times_{\rB G}\rB H)
    \end{equation}
\end{notation}

\iffalse
\begin{lemma}\label{lemma:doublecosetass}
   Let $G$ be an affine algebraic group over characteristic $0$ field, and $H,H'\subset G$ closed subgroups.  Then $H\backslash G/H'$ satisfies the conditions of \Cref{ass:qcadft}.  
\end{lemma}

\begin{proof}\dev{check this proof}
  The morphism $\rB H\to \rB G$ is quasi-compact, so $\rB H\times_{\rB G} \rB H'$ is quasi-compact.  The result follows from Lemmas \ref{lemma:affdiag} and \ref{lemma:finitetype}. 
\end{proof}
\fi

\begin{example}
   For $G$ an affine algebraic group over characteristic $0$ field and $H,H'\subset G$ closed subgroups the following equivalence follows from \Cref{corr:IntegralFunctor}.  
   \begin{equation}
       \Fun^\rL_{\nRep(G)}(\nRep(H),\nRep(H'))\simeq \nQcoh^\star(\rB H\times_{\rB G}\rB H')\simeq \nQcoh^\star(H\backslash G/H'). 
   \end{equation}
  There is a standard formula for $H\backslash G/H'$ given by:
\begin{equation}\label{eq:doublecoset}
     H\backslash G/ H' = \bigsqcup_{[g]\in \pi_0(H\backslash G/H')} \rB(H\cap gH'g^{-1})\,.
\end{equation}
The notation $\pi_0(H\backslash G/H')$ indicates the 0-homotopy component, which recovers the ordinary quotient, rather than the stacky double coset.  
%here are still ``$\nRep$'' contributions from points that are stabilized by the conjugation action.
\end{example}

We now give one more example of physical interest involving a specific choice of $G$.  Note that, over $\mathbb{C}$, we have $\mathbb{G}_m=\mathbb{C}^\times$ and $\mu_n=\mathbb{Z}_n$. 
\begin{example}
Let $G=\mathbb{G}_m$ and $H=\mu_j$, $H'=\mu_k$.  Then we can compute 
\begin{align}
    \Fun^\rL_{\nRep(\mathbb{G}_m)}(\nRep( \mu_j),\nRep ( \mu_k))&= \nQcoh^\star( \mathbb{G}_m\times \mu_{\mathrm{gcd}(j,k)})\\ \notag 
    &\simeq \nQcoh^\star( \mathbb{G}_m)\otimes_{\nVect}\nRep( \mu_{\mathrm{gcd}(j,k)})
\end{align}
For $j=k$, one obtains 
\begin{equation}
\End_{\nRep(\mathbb{G}_m)}(\nRep( \mu_k))= \nQcoh^\star( \mathbb{G}_m)\otimes_\nVect \Rep( \mu_k)\,.
\end{equation}
\end{example}

\begin{lemma}\label{lemma:HGDescent}
   Let $G$ be an affine algebraic group over characteristic $0$ field, and $H\subset G$ a closed subgroup.  Then the inclusion $\rB H\to \rB G$ satisfies descent.  
\end{lemma}

\begin{proof}
We wish to confirm that $\rB H\to \rB G$ is faithfully flat, and so satisfies descent by \cite[Example 4.6]{StefanichSheaf}.  To do so it is sufficient to confirm that base change $G/H=\rB H\times_{\rB G}\Spec({\mathbb{K}})$ along the cover $\Spec(\mathbb{K})\to \rB G$ is faithfully flat over $\Spec({\mathbb{K}})$. But we know that $G/H\to \Spec({\mathbb{K}})$ is faithfully flat since $G/H$ is non-empty and $\mathbb{K}$ is a field.  
\end{proof}

The following result is a specialization of \Cref{thm:MoritaEquivalence}, formulated in a form that is directly applicable to the physical setting considered in this section.
\begin{proposition}\label{corr:GroupMorita}
   Let $G$ be an affine algebraic group over characteristic $0$ field $\mathbb{K}$, and $H\subset G$ a closed subgroup.  Suppose that \underline{\textbf{either}}: $n=1$ and $H,G$ are reductive \underline{\textbf{or}} $n\geq 2$.  Then there is a Morita equivalence 
   \begin{equation}\label{eq:moritagroups}
       \Fun_{\nRep(G)}(\nRep(H),-):\Mod_{\nRep(G)}\xrightarrow{\sim}\Mod_{\nQcoh^\star(H\backslash G/H)}
   \end{equation}
   where $\nQcoh^\star(H\backslash G/H)=\nQcoh^\star(\rB H\times_{\rB G}\rB H)\simeq \End^\rL_{\Rep(G)}(\Rep(H))$.  
\end{proposition} 

%Since $G$ is affine over a field of characteristic 0, it is an algebraic subgroup of some $\GL_n$ \cite[Corollary 4.10]{Milne_2017}, and hence finite type.  

\begin{proof}

$\rB G$ is $n$-affine for all $n\geq 1$ by \Cref{cor:BGnAff}.  When $n\geq 2$, the result follows from \Cref{thm:MoritaEquivalence} and \Cref{lemma:HGDescent}.  In the case $n=1$, we must additionally confirm that $\rB H\to \rB G$ is affine schematic.  

It is sufficient to confirm that the base change along an affine cover of $\rB G$ is affine.  We consider base change of the cover $\Spec(\mathbb{K})\to \rB G$.  The result follows since $G/H\simeq \Spec(\mathbb{K})\times_{\rB G}\rB H$ is affine by Matsushima criterion.    
\end{proof}

\begin{rem}
    The previous result can likely be extended beyond affine groups using \cite[Proposition 11.2.1]{Gaitsgory_2015}. 
\end{rem}

\begin{example}
For $H=H'=\{1\}$, one gets 
\begin{equation}
\End^\rL_{\nRep(G)}(\nVect,\nVect)\simeq \nQcoh^\star(\pt\times_{\rB G}\pt)\simeq \nQcoh^\star(G), 
\end{equation}
This gives an analog of the usual Morita equivalence between $\Rep(G)$ and $\Vect(G)$ in arbitrary dimension. 
\begin{equation}
   \Fun_{\nRep(G)}(\nVect,-):\Mod_{\nRep(G)}\xrightarrow{\sim}\Mod_{\nQcoh^\star(G)}.  
\end{equation}

\end{example}
 
\begin{corollary}
    There is an equivalence of Drinfeld centers 
    \begin{equation}
         \FZ(\nQcoh^\star(H\backslash G/H))\simeq \FZ(\nRep(G)). 
    \end{equation}
    In particular, we get
     \begin{equation}
        \FZ(\nQcoh^\star(G))\simeq \FZ(\nRep(G)). 
    \end{equation}
\end{corollary}
\begin{proof}
    The results follow from \cite[Theorem 5.3.1.30]{LurieHA} and \Cref{corr:GroupMorita}. 
\end{proof}

We defer the definition of the Drinfeld center to \S\ref{subsection:electricsym}. One can recall, by \Cref{eq:codim2}, that the codimension-2 operators in this bulk SymTFT are labeled by objects in $\FZ(\nRep(G))$.
The previous two results can be interpreted physically as the fact  that, for all closed subgroups $H\subset G$, the categories $\nQcoh^\star(H\backslash G/H)$ label boundary conditions for the same $(n+2)$-dimensional SymTFT.  The following proposition gives a way to calculate this Drinfeld center explicitly. 

\begin{notation}
    For an algebraic group $G$, the free loop space $\mathcal{L}\rB G=\rB G\times_{\rB G\times \rB G}\rB G$ is identified with the (stacky) \emph{conjugation quotient} $G/G$.  In the following, the symbol $G/G$ will refer to the stack quotient in this sense.
\end{notation}

\begin{proposition}\label{prop:center}
     Let $G$ be an affine algebraic group over characteristic $0$ field, and $H\subset G$ a closed subgroup. There is an equivalence of braided categories
   \begin{equation}
       \FZ(\nQcoh^\star(H\backslash G/H))\simeq \nQcoh^\star(G/G)\simeq \FZ(\nRep(G)). 
   \end{equation}
\end{proposition}
   
\begin{proof}
The result follows from \Cref{Thm:BFNCenter} and \Cref{lemma:HGDescent}. 
\end{proof}

The adjoint conjugation quotient stack $G/G$ has a presentation 
\begin{equation}
   \mathcal{L}\rB G\simeq G/G\simeq \bigsqcup_{[g]\in \pi_0(G/G)}\rB C_G(g). 
\end{equation}
Where $\pi_0(G/G)$ is the collection of \emph{conjugacy classes} of $G$ and $C_G(g)$ is the \emph{centralizer} of a representative of the chosen conjugacy class $[g]\in G/G$.  Using that $\nQcoh$ takes colimits to limits, the above formula gives a decomposition of $\FZ(\nRep(G))$ into a product of representations of centralizers over conjugacy classes. 
\begin{equation}\label{eqn:DrinfeldDecomp}
   \FZ(\nRep(G))\simeq \nQcoh^\star\left(\bigsqcup_{[g]\in \pi_0(G/G)}\rB C_G[g]\right)\simeq \prod_{[g]\in \pi_0(G/G)}\nRep(C_G[g])\,.
\end{equation}

%\begin{example}
%Let $G=\mathbb{G}_m$.  Then can compute. \matt{finish}
%\end{example}

With these details in place, we can now finish the proof of \Cref{mainthm:Wilson}.

\begin{proposition}\label{lemma:thmBpart1}
Let $G$ be an affine algebraic group over characteristic $0$ field, and $H\subset G$ a closed subgroup.  There is a functor $\rB^{n}\Rep(H)\to \Mod_{\nRep(G)}$.
   % $\rB^{n}\Qcoh^\star(H\backslash G/H)\to \Mod_{\nRep(G)}$. 
\end{proposition}

\begin{proof}
The $\nRep(G)$-module $\nRep(H)$ has symmetries given by $$\End^\mathrm{L}_{\nRep(G)}(\nRep(H)) \simeq \nQcoh^\star(H\backslash G/H)\,,$$ thanks to \Cref{corr:GroupMorita}. Using \Cref{eq:doublecoset} and inspecting the identity conjugacy class in $H\backslash G/ H$, leads to an inclusion $\nRep(H)\rightarrow \nQcoh^\star(H\backslash G/H)$. Using the Morita equivalence in \Cref{eq:moritagroups} and delooping, we obtain a functor $\rB \nRep(H)\rightarrow \Mod_{\nRep(G)}$. Using affineness gives the identification $\nRep(H)\simeq  \Mod^{n-1}_{\Rep(H)}$, from which we obtain the following sequence of functors
\begin{equation}
    \rB^n \Rep(H)\rightarrow \rB \Mod^{n-1}_{\Rep(H)} \simeq \rB \nRep(H) \rightarrow \Mod_{\nRep(G)}.
\end{equation}
% Since $H\backslash G/H$ is $1$-affine (and hence $\geq 1$ affine) by \cite[Theorem 2.2.4]{Gaitsgory_2015}, we can in particular write a monoidal functor
%\begin{equation}
 %   \rB^{n-1}\Qcoh^\star(H\backslash G/H)\to \Mod^{n-1}_{\Qcoh^\star(H\backslash G/H)}\simeq \nQcoh^\star(H\backslash G/H)\simeq \End^\mathrm{L}_{\nRep(G)}(\nRep(H)).
%\end{equation}
%By delooping the above sequence, we obtain a functor $\rB^n \Qcoh^\star(H\backslash G/H) \to \Mod_{\nRep(G)}$, which is the analog of \Cref{eq:0form}. 
\end{proof}
%The functor $\rB^n \Qcoh^\star(H\backslash G/H) \to \Mod_{\nRep(G)}$ factors through $\rB \End^\mathrm{L}_{\nRep(G)}(\nRep(H))$. 

In particular since $\End^\mathrm{L}_{\nRep(G)}(\nRep(H))$ is the category of all symmetries with respect to the $\nRep(G)$-module $\nRep(H)$, \Cref{lemma:thmBpart1} exhibits a $(n-1)$-form  $\Rep(H)$ symmetry on $\nRep(H)$ by the external perspective in \Cref{ex:zeroform}. One can view this as giving an inclusion of codimension-$n$ Wilson line operators of the $(n+1)$-dimensional boundary associated to $\nRep(H)$. If $H = G$, which is to say that there is no symmetry breaking, we get the expected $(n-1)$-form  $\Rep(G)$ symmetry.

%%%%%%%%%%%%%%%%%%%%%%%%%%%%%%
\subsubsection{The 1-form center symmetry}
 %%%%%%%%%%%%%%%%%%%%%%%%%%%%%%%%%%%%%%%%

We now present how to observe the 1-form center symmetry for $G$-gauge theory on the boundary. From the perspective of the bulk TQFT, the objects that implement this symmetry link with the $(n-1)$-form $\Rep(H)$ symmetry found in the previous section.

\begin{lemma}\label{lemma:Hfiber}
Let $G$ be an affine algebraic group over characteristic $0$ field and $H$ a closed subgroup.  There is a pullback square  
   \begin{equation}
     \begin{tikzcd}[row sep= small, column sep= 3cm]
H \ar[r]\ar[d]\ar[dr, phantom, "\lrcorner", pos=0.18] & \mathrm{pt} \ar[d] \\
\rB H\times_{\rB H\times \rB G}\rB H \ar[r,"\delta=\id_{\rB H}\times _{\pi_{\rB G}} \id_{\rB H}"'] & \rB H\times_{\rB G}\rB H. 
\end{tikzcd}
   \end{equation} 
\end{lemma}
\begin{proof}
Write $\pt=\pt\times_\pt \pt$.  Since limits commute, the pullback of the above diagram is given by 
\begin{equation}
    (\rB H\times_{\rB H}\pt)\times_{(\rB H\times \rB G)\times_{\rB G}\pt}(\rB H\times_{\rB H}\pt)\simeq \pt\times_{\rB H}\times \pt\simeq H.
\end{equation}
\end{proof}

\begin{theorem}\label{thm:fiber}
    Let $G$ be an affine algebraic group over characteristic $0$ field and $H\subset G$ a closed subgroup.  There is a an injective map 
    \begin{equation}
        \rZ(H)\to \nVect \otimes_{\nQcoh^\star(H\backslash G/H)}\FZ(\nQcoh^\star(H\backslash G/H)), 
    \end{equation}
    of the center $\rZ(H)$ into the fiber of $\FZ(\nQcoh^\star(H\backslash G/H))\to \nQcoh^\star(H\backslash G/H)$. \end{theorem}
    
\begin{proof}
We will make use of the characterization of the map $\FZ(\nQcoh^\star(H\backslash G/H))\to \nQcoh^\star(H\backslash G/H)$ given in \Cref{Thm:BFNCenter}, where the morphisms $\delta,\pi$ are defined in \Cref{eqn:deltapi}.  We will proceed by first obtaining maps $\rZ(H)\to \fib(\delta_*)$ and $\rZ(H)\to \nQcoh(\mathcal{L}\rB G)$, and then showing that this lifts to the desired fiber given by $\fib(\delta_* \pi^*)$ in the following double pullback square:  
\begin{equation}
    \begin{tikzcd}
\fib(\delta_*\pi^*)\ar[d]\ar[r]\ar[dr,phantom,very near start, "\lrcorner"]&\fib(\delta_*)\ar[d]\ar[r]\ar[dr,phantom, very near start, "\lrcorner"]&\nVect\ar[d]\\
       \nQcoh^\star(\mathcal{L}\rB G) \ar[r,swap,"\pi^*"]&\nQcoh^\star(\rB H\times _{\rB H\times \rB G}\rB H)\ar[r,swap,"\delta_*"]&\nQcoh^\star(H\backslash G/H)\,.
    \end{tikzcd}
\end{equation}
There is a map $\rZ(H)\to \nQcoh^\star(\mathcal{L}\rB G)$ taking $z\in \rZ(H)$ to the skyscraper sheaf supported on the conjugacy class specified by $z$. 

The stacks $\rB H\times_{\rB H\times \rB G}\rB H$ and $H\backslash G/H$ are $n$-affine by \Cref{cor:BGnAff}.  By Proposition \ref{Prop:BFN} and Lemma \ref{lemma:Hfiber}, we can establish the following  equivalence. 
\begin{align}
\fib(\delta_*)&=\nVect\otimes_{\nQcoh^\star(H\backslash G/H)}\nQcoh^\star(\rB H\times_{\rB H\times \rB G}\rB H) \\\notag 
   &\simeq \nQcoh^\star(\pt \times _{\rB H\times_{\rB G}\rB H}(\rB H\times_{\rB H\times \rB G}\rB H))\\ \notag 
   &\simeq \nQcoh^\star(H). 
\end{align}
  Then there is a map $\rZ(H)\to \nQcoh^\star(H)\simeq \fib(\delta_*)$, taking $z\in \rZ(H)$ to the skyscraper sheaf supported at $z$.  To give a map $\rZ(H)\to \fib(\delta_*\pi^*)$ using the universal property, it is sufficient to show that the image of the two maps $f$ and $g$ coincide on  $\nQcoh^\star(\rB H \times_{\rB H\times \rB G}\rB H)$ in the following diagram:
\begin{equation}
    \begin{tikzcd}[row sep=small]
	{\mathrm{Z}(H)} \\[-10pt]
	& \fib(\delta_*\pi^*) & {\mathbf{nQCoh}^\star(H)} \\
    & &\\
	& {\mathbf{nQCoh}^\star(\mathcal{L}\mathrm{B}G)} & {\nQcoh^\star(\rB H \times_{\rB H\times \rB G}\rB H)\,.}
	\arrow[dotted, from=1-1, to=2-2]
	\arrow["f", curve={height=-24pt}, from=1-1, to=2-3]
	\arrow["g"', curve={height=10pt}, from=1-1, to=4-2]
	\arrow[from=2-2, to=2-3]
	\arrow[from=4-2, to=4-3]
	\arrow[from=2-2, to=4-2]
	\arrow[from=2-3, to=4-3]
\end{tikzcd}
\end{equation} 
 % \begin{gather}
  %   \rZ(H)\to \nQcoh(H)\to \nQcoh(BH_{BH\times BG}BG) \\
   %  \rZ(H)\to \nQcoh(\mathcal{L}BG)\xrightarrow{\pi^*} \nQcoh(BG_{BH\times BG}BG)
  %\end{gather}
    Since the map $\rB H\times _{\rB H\times \rB G}\rB H\to \rB H\times_{\rB G\times \rB G}\rB H\simeq \mathcal{L}\rB G$ is injective, we see that $\pi^*$ takes skyscraper sheaves on the conjugacy class $[z]\in \mathcal{L}\rB G$ to skyscraper sheaves on $[z]\in \rB H\times_{\rB H\times \rB G}\rB H$.

  The map $H\to \rB H\times_{\rB H\times \rB G}\rB H$ takes $z\in \rZ(H)$ to the conjugacy class $[z]$, and the map $\nQcoh^\star(H)\to \nQcoh^\star(\rB H\times_{\rB H\times \rB G}\rB H)$ is given by pushforward along this inclusion.  Then it takes skyscraper sheaves on $z\in \rZ(H)$ to skyscraper sheaves on $[z]\in \rB H\times_{\rB H\times \rB G}\rB H$, so the images of the two maps starting at $\rZ(H)$ and ending at  $\nQcoh^\star(\rB H\times_{\rB H\times \rB G}\rB H)$ coincide.  By the universal property of the fiber, this lifts to the desired map $\rZ(H)\to \fib(\delta_*\pi^*)$, which is identified with the fiber of $\FZ(\nQcoh^\star(H\backslash G /H))\to \nQcoh^\star(H\backslash G/H)$ by \Cref{Thm:BFNCenter}. 
\end{proof}

The objects in $\FZ(\nQcoh^\star(H\backslash G/ H))$ are codimension-2 in the $(n+2)$-dimensional TQFT defined in \Cref{thm:generalTQFT}. Being in the fiber of the map $\FZ(\nQcoh^\star(H\backslash G/ H))\rightarrow \nQcoh^\star(H\backslash G/ H)$  means that the objects end on the boundary labeled by the $\nRep(G)$-module $\nRep(H)$. Therefore they correspond to codimension-2 objects on the boundary, with labels in $\rZ(H)$. Finally, if $H=G$ i.e. when no symmetry is broken, then the boundary contains $\rZ(G)$ 1-form symmetry operators. This finishes the proof of \Cref{prop:linkingelectric}.

\begin{rem}
    Flat $G$-gauge theory should also have $1$-form symmetry labeled by other conjugacy classes which meet $H$.  These do not appear in the fiber since they are not trivialized on the boundary.  Rather, they become \emph{direct integral} objects, which include the vacuum line as a `summand'.  These direct integral objects have an ending operator on the boundary given by the projection to the vacuum line. 
\end{rem}

Having recovered the center symmetry and its breaking, we now invoke \Cref{corr:GroupMorita}. This result implies that the symmetry data of pure gauge theory and of gauge theories with additional matter content admit a unified description within a single bulk TQFT.  This provides another piece of evidence in support of \Cref{claim:mainclaim}.

%%%%%%%%%%%%%%%%%%%%%%%%%%%%%%%%%%
\subsection{Magnetic Symmetry}\label{subsection:twistedsectors}
%%%%%%%%%%%%%%%%%%%%%%%%%%%%%%%%%%

A module category for $\Rep(G)$ should be interpreted as a topological boundary conditions for the SymTFT as in \Cref{fig:2dSymTFT}.  Before discussing magnetic symmetries in general dimensions, we first turn to the classification of invertible modules of $\Rep(G)$.  This also provides the technical details of \Cref{lemma:invertiblemodule}. We will then lift the $\Rep(G)$-module categories to morphisms in higher categories, and interpret them as the magnetic symmetries exhibited by flat $G$-gauge theories\footnote{We call these ``magnetic symmetries'' because of the close relationship with magnetic symmetries of Yang--Mills theories discussed in \Cref{rem:MagneticYM}.}.   The following construction yields invertible modules for $\Rep(G)$.  Consider the functor 
\begin{align}
  \mathcal{F}:\CAlg\to \Spaces,
\end{align}
which associates to each algebra $A$ the space $\Mod_{\Mod_A}^\times$ of invertible $\Mod_A$-module categories.  $\mathcal{F}$ is a sheaf for the étale topology, and so defines an étale stack.  As a stack, this is the \emph{classifying stack for invertible module categories}.  

A point in the stack $\mathcal{F}$ is given by a point in the space $\mathcal{F}(\Spec(\mathbb{S}))$, where $\mathbb{S}$ is the sphere spectrum. $\mathcal{F}$ has a natural pointing given by $\Mod_{\mathbb{S}}\in \Mod_{\Mod_\mathbb{S}}^\times$.  It follows from \cite[Corollary 5.4.4]{StefanichLinear} that the pointed stack $\cF$ is equivalent to $\rB^2\GL_1$ as a pointed stack. 

\begin{construction}\label{contruction:invertible}
   Let $H$ be an affine algebraic group such that $\rB H$ is $1$-affine.  Then $\mathcal{F}(\rB H)\simeq \Mod_{\Rep(H)}^\times$, and we note the following equivalence: 
\begin{equation}
\rH^2(\rB H;\mathbb{C}^\times)\simeq \rH^2(\rB H,\GL_1)\simeq \Hom^{E_0}(\rB H,\rB^2\GL_1)\simeq \Hom^{E_0}(\rB H,\mathcal{F})\simeq \Mod_{\Rep(H)}^\times\,.
\end{equation}
For every $\alpha \in \rH^2(\rB H;\mathbb{C}^\times)$, denote its image in $\Mod_{\Rep(H)}^\times$ by $\Rep^\alpha(H)$. The analog of these categories in the fusion case were studied in \cite{ostrik2003module}. Similar results for the category of coherent sheaves over an affine group scheme have been obtained in \cite{GelakiAffine}.  By \cite[Corollary 5.4.4]{StefanichLinear} it follows from  that one can construct invertible $\Rep(H)$-module category in this way, and for $\alpha,\beta\in \rH^2(\rB H;\mathbb{C}^\times)$, we have fusion rules 
\begin{equation}
\Rep^\alpha( H)\otimes_{\Rep(H)}\Rep^{\beta}(H)\simeq \Rep^{\alpha+\beta}( H). 
\end{equation}
\end{construction}

\begin{lemma}\label{prop:moritaforn}
Let $G$ be an affine algebraic group over characteristic $0$ field, $H\subset G$ a subgroup, and $\alpha\in \rH^2(\rB H;\mathbb{C}^\times)$.  Then $\Rep^\alpha(H)$ is dualizable in $\Mod_{\Rep(G)}$.
\end{lemma}
\begin{proof}
    We note that by \cite[Theorem 2.2.2]{Gaitsgory_2015}, $\rB G$ is $1$-affine.  Then $\Rep(H)$ is canonically self-dual in $\Mod_{\Rep(G)}$ by \cite[Corollary 4.8]{ben2010integral}.  The induction map is implemented by 
   \begin{equation}
       \Rep(H)\otimes_{\Rep(H)}(-):\Mod_{\Rep(H)}\to \Mod_{\Rep(G)}, 
   \end{equation}
   and so this map takes the invertible $\Rep(G)$-modules $\Rep^\alpha(H)$ to dualizable objects in $\Mod_{\Rep(G)}$.
\end{proof}

We now present several examples and the corresponding fusion rules for the module categories, along with computations of their endomorphism categories.  We remark that these results parallel those of the fusion category case, and extend naturally to groups with richer structure.

\begin{example}\label{ex:dyonic}
Let $G$ be an affine algebraic group over characteristic $0$ field, and let $H=H'=G$ and $\alpha,\beta\in \rH^2(\rB G;\mathbb{C}^\times)$. One sees that
\begin{align}
\Fun^\rL_{\Rep(G)}(\Rep^\alpha(G),\Rep^\beta( G))&\simeq \Rep^{-\alpha}( G)\otimes_{\Rep(G)}\Rep^{\beta}(G)\\ \notag
&\simeq \Rep^{\beta-\alpha}( G).  
\end{align}
In particular,
\begin{equation}
    \Fun^\rL_{\Rep(G)}(\Rep^{\alpha}( G),\Rep^{\alpha}(G))\simeq \Rep(G)\,.
\end{equation}
\end{example}
Operators in $\Rep(G)$, considered as endomorphisms of $\Rep^\alpha(G)$, can be interpreted as Wilson lines, while the operators in $\Rep^{\beta-\alpha}(G)$, considered as morphisms from $\Rep^\alpha(G)$ to $\Rep^\beta(G)$, should be interpreted as \emph{dyonic} operators with magnetic charge $\alpha-\beta$.

\begin{proposition}
    There is a functor $\rB^{n-1} (\rH^2(\rB G; \mathbb{G}_m) )\rightarrow \Mod_{\nRep(G)}$, where $\rB^{n-1} (\rH^2(\rB G; \mathbb{G}_m))$ is the classifying space of a $(n-2)$-form  $\rH^2(\rB G; \mathbb{G}_m)$ symmetry.
\end{proposition}

\begin{proof} Consider the case when $n=1$ in which the category of symmetries is $\Mod_{\Rep(G)}$. To get to general $n$ dimensions we apply  $\Mod^{n-1}$ onto $\Mod_{\Rep(G)}$, and this gives $\Mod_{\nRep(G)}$. In particular, we have used the fact that since $\rB G$ is $1$-affine, by \Cref{cor:BG1affine}, we have $\nRep(G) = \Mod^{n-1}_{\Rep(G)}$. By applying $\Mod^{n-1}$ also to $\rH^2(\rB G; \mathbb{G}_m)$ in \Cref{eq:maglabels} we get a functor
\begin{equation}
    \mathbf{Mod}^{n-1}(\rH^2(\rB H;\mathbb{C}^\times)) \rightarrow \Mod_{\nRep(G)}
\end{equation}
from which we have a functor
\begin{equation}
  \rB^{n-1} (\rH^2(\rB H;\mathbb{C}^\times))\rightarrow \Mod^{n-1}(\rH^2(\rB H; \mathbb{C}^\times))\rightarrow \Mod_{\nRep(G)}\,.
\end{equation}
 \end{proof}
This result guarantees that the same construction/family of invertible $\Rep(G)$-module categories also exhibits a magnetic symmetry in higher dimensions. Furthermore, we see that the magnetic symmetry  is generated by a surface operator, as explained in  \S\ref{subsection:symgauge}.

We now use the contents of \S\ref{section:Equivariantization}, applied to the stack $\rB G$, to obtain the following higher analog of the fusion-categorical result \cite[Lemma 1.3.8]{decoppet2022morita}.  This will allow us to lift the magnetic symmetry operators to higher dimensions.  In the following statement, $\underline{\nVect}$ refers to the \emph{geometric category} of vector spaces, which is defined in \Cref{ex:geometricnVect}. 

\begin{corollary}\label{cor:BG1affine}
Let $G$ be an affine algebraic group over characteristic $0$ field.  There is an equivalence
\begin{equation}
 {\nRep}(G)\simeq {\Fun}(\rB G,\underline{\nVect})\simeq\Mod^{n-1}_{{\Rep}(G)}. 
\end{equation}
\end{corollary}

\begin{proof}
$\rB G$ $j$-affine for $j\geq 1$ by \Cref{cor:BGnAff}.  Using \Cref{lem:affinetomod} with $X=\rB G$ and $Y =\Spec(\mathbb{K})$ gives the result. This also establishes \Cref{prop:equiv}.
\end{proof}

%%%%%%%%%%%%%%%%%%%%%%%%%%%%%%
\section{Main proofs}\label{section:technicalproofs}
%%%%%%%%%%%%%%%%%%%%%%%%%%%%%%
In this section we give the proofs of the mathematical facts that are needed to obtain the physical results in \S\ref{section:symTFT} and \S\ref{section:dynamicalgauging}. In particular the main results in this section are the following:
\begin{itemize}
\item We prove the general form of \Cref{cor:BG1affine}, pertaining to equivariantization, in \Cref{lem:affinetomod}. The result  was used to construct the category $\nRep(G)$ for the purpose of defining bulk TQFTs in \S\ref{subsection:buldTFT} and \S\ref{subsection:generalTQFT}.
\item We define and compute the Drinfeld center for the category of higher quasi-coherent sheaves. This leads to \Cref{Thm:BFNCenter}.
\item We prove a generalized version of \Cref{corr:GroupMorita} in \Cref{thm:MoritaEquivalence}.
\end{itemize}

\subsection{Geometric Categories}\label{subsection:geometriccats}
The category $\Stk$ is particularly well behaved: it is an $\infty$-topos \cite[Proposition 6.2.2.7]{LurieHTT}.  This means that it serves as a good replacement for the category $\Spaces$.  In particular, we can do $\infty$-category theory \emph{internally} to $\Stk$ \cite{martini2022yonedas,martini2024colimits,martini2025presentability,2Topoi}.  This gives the notion of \emph{geometric categories}: categories which have geometric spaces of objects and morphisms, and whose categorical structure respects this geometry\footnote{This data must also include higher morphisms, making it a \emph{simplicial object} in $\Stk$, and satisfy \emph{Segal} and \emph{Rezk completness} conditions.}.  

\begin{definition}\label{def:valuedsheaves}
    Let $\cD$ be an $\infty$-category.  A \emph{$\cD$-valued sheaf} on $\Stk$ is a limit-preserving functor $\Stk^{\op}\to \cD$.  Denote by $\Sh_\cD(\Stk)\subset \Fun(\Stk^{\op},\cD)$ the full subcategory on the $\cD$-valued sheaves.  
\end{definition}
A geometric category is equivalently characterized as a sheaf of categories on $\Stk$ \cite[Section 3.5]{martini2022yonedas}, or as a functor $\Aff^{\op}\to \Cat_\infty$ which satsifies descent with respect to the chosen notion of cover \cite[Proposition 1.1.12]{LurieDAGV}.  The theory of internal $(\infty,n)$-categories has not yet been fully developed, but nonetheless, we proceed analogously. 
\begin{definition}
   The category of \emph{geometric} $n$-categories is the full subcategory 
   \begin{equation}
    \mathbf{n}\Cat_\Stk\subset \Fun(\Stk^{\op},\Cat_{(\infty,n)})
\end{equation}
on the limit preserving functors.  We analogously define geometric (presentable/stable/monoidal) $n$-categories as sheaves of (presentable/stable/monoidal) $n$-categories on $\Stk$, potentially after universe enlargement.\footnote{Presentable and stable internal categories are studied in \cite{martini2025presentability}.} To use the terminology of \Cref{def:valuedsheaves}, these are $n$-category valued sheaves on $\Stk$.
\end{definition}
 %\dev{check that the definitions agree}

\begin{example}\label{ex:geometricnVect}
The functor 
   \begin{equation}
       \nQcoh:\Stk^{\op}\to \CAlg(\nPr^\mathrm{L}_\st)
   \end{equation} 
   defines a symmetric monoidal, presentable, stable \'etale geometric $n$-category by \cite[Corollary 14.3.5]{stefanich:thesis}, and an fppf geometric $n$-category\footnote{With the same adjectives.} by \cite[Example 4.6]{StefanichSheaf}, which we denote by $\underline{\nVect}$. 
One should think of this as the \emph{geometric category of $n$-vector spaces}.  By this, we mean that its objects are $n$-vector spaces, but the geometric category $\underline{\nQcoh}$ additionally remembers how to smoothly deform between $n$-vector spaces. 
   
\end{example}

\begin{example}
The functor
    \begin{equation}
        \nQcoh^\otimes(X\times -):\Aff^{\op}\to \CAlg(\nPr^{\mathrm{L}}_\st), 
    \end{equation} 
    for any $X\in \Stk$, defines a symmetric monoidal, presentable stable (\'etale/fppf) geometric $n$-category, which we denote by $\underline{\nQcoh}^\otimes(X)$.  One should think of this as the geometric category which encodes families of $n$-vector spaces parameterized by $X$.  
\end{example}

\begin{example}
    Let $G$ be an algebraic group.  One should think of $\underline{\Qcoh^\otimes}(\rB G)$ as the \emph{geometric category of $G$-representations}. 
\end{example}
Let $\underline{\Fun}$ denote the internal hom for geometric $n$-categories. 
\begin{lemma}\label{lemma:QCohXY}
Consider $X,Y\in \Stk$.  There is a symmetric monoidal equivalence 
\begin{equation}
\underline{\Fun}(X,\underline{\nQcoh}^\otimes(Y))\simeq \underline{\nQcoh}^\otimes(X\times Y)\,.
\end{equation}
\end{lemma}

\begin{proof}
$\underline{\nQcoh}^\otimes(X)$ defines a limit preserving functor $\nQcoh:\Stk^{\op}\to \nPr^\mathrm{L}_\st$.  In particular, it takes tensors in $\Stk$ to cotensors in $\CAlg(\mathbf{nPr}^\rL_\st)$.  The tensor of $X$ on $Y$ is given by $X\times Y$, and the cotensor of $X$ on $\underline{\nQcoh}^\otimes(Y)$ is given by $\underline{\Fun}(X,\underline{\nQcoh}^\otimes(Y))$, giving the desired equivalence. 
\end{proof}

\subsection{Equivariantization}\label{section:Equivariantization}
%%%%%%%%%%%%%%%%%%%%%%%%%%%%%%%%%%

We now prove a generalization of the fusion-categorical equivariantization procedure developed in \cite{DGNO}.  This makes use of the external perspective on symmetries presented in \S\ref{subsection:external}, and will serve as the first step in our gauging procedure. 

%\dev{I can make this an equivalence of geometric categories too}

\begin{warning}
    The following are \textbf{not} equivalances of geometric categories, only of underlying presentably monoidal $(\infty,n)$-categories.  It would be interesting to extend these results to be fully geometric.  We do not do so here, since internal higher algebra is not sufficiently developed at the time of writing.  We refer to \cite[Section 2.5]{martini2025presentability} for the current state of the art, and \cite{martini2022yonedas,martini2022cocartesian,martini2024colimits, 2Topoi} for other relevant constructions in internal categories. 
\end{warning}

\begin{lemma}\label{lem:affinetomod}
 Let $X\in \Stk$ be an $m$-affine stack.  There is an equivalence 
    \begin{equation}
    \Fun(X, \underline{\nVect})\simeq \Mod^{n-m}_{\mathbf{m}{\Qcoh}(X)}\,.
    \end{equation}
\end{lemma}
\begin{proof}
      We get a sequence of equivalences 
      \begin{equation}
          {\Fun}(X,\underline{\nVect})\simeq {\nQcoh}^\otimes(X)\simeq \Mod_{(\mathbf{n-1}){\Qcoh}^\otimes(X)}\simeq \hdots \simeq \Mod^{n-m}_{{\mathbf{m}\Qcoh^\otimes}(X)}\,, 
      \end{equation}
     by, respectively, using Lemma \ref{lemma:QCohXY} and passing to underlying (non-geometric) categories, and the fact that if $X$ is $m$-affine, then it is $j$-affine for $j\geq m$. 
\end{proof}

%\matt{this title is nice, but if we are to keep it, it would be  good to expand on why its an integral transform. Maybe the title should have dinfeld center in it}
%%%%%%%%%%%%%%%%%%%%%%%%%%%%%%%%%%
\subsection{The Drinfeld Center
}\label{subsection:electricsym}
%%%%%%%%%%%%%%%%%%%%%%%%%%%%%%%%%%
We now build upon the definitions developed earlier in this section to introduce additional structures from higher sheaf theory. In particular, we proceed to define and compute the Drinfeld center for higher categories of quasi-coherent sheaves. Subsequently, we construct the corresponding bulk-to-boundary map for the center. These results culminate in the main result of this section, stated in \Cref{Thm:BFNCenter}. The work in this section also sets the stage for the computations that are done in \S\ref{subsection:extractingelectric} to explicitly study electric 1-form symmetries.

The Drinfeld center plays a key role in the fusion categorical approach to gauging.  We will require a generalization to the $\infty$-categorical setting.  The following is \cite[Definitions 5.3.1.2, and 5.3.1.12]{LurieHA}.
\begin{definition}\label{def:drinfeld}
Let $\cC$ be a (presentable/stable/enriched) monoidal $\infty$-category.  The \emph{Drinfeld center} of $\cC$ is terminal amongst (presentable/stable/enriched) monoidal categories $\cA$ equipped with a morphism $\mathcal{A}\otimes \mathcal{C}\to \mathcal{C}$ making the following diagram commute
\begin{equation}\label{eqn:deltapi}
        \begin{tikzcd}
& \mathcal{A}\otimes \mathcal{C}\ar[dr]&\\
\mathcal{C}\ar[ur,"\iota_\mathcal{    A}\otimes \id_\mathcal{C}"]\ar[rr,swap,"\id_\cC"]&&\mathcal{C},
\end{tikzcd}
\end{equation}   
where $\iota_\mathcal{A}:\mathds{1}\to \mathcal{A} $ is the unit morphism.  The Drinfeld center exists and inherits a natural braided monoidal structure by \cite[Corollary 5.3.1.15]{LurieHA}, and will be denoted by $\FZ(\cC)$.
\end{definition} 

\begin{notation}Consider $f:X\to Y$ in $\prstk$.  Inductively denote 
   \begin{align}
 X^{\times^{i+1}_Y}:=X\times_Y (X^{\times^i_Y})&&(X\times_Y X)^{\times^{i+1}_X}:= (X\times_Y X)\times_X (X\times_Y X)^{\times^i_X}.      
   \end{align} 
\end{notation}

\begin{lemma}\label{lemma:BFNPullback}
  The following is an equivalence in $\prstk$: 
    \begin{equation}
    X\times_{X\times X}((X\times_YX)^{\times^{i+1}_X})\simeq \mathcal{L}Y\times_Y X^{\times^{i+1}_Y}\,.
\end{equation}
\end{lemma}

\begin{proof}
The case $i=0$ follows from the following diagram and pasting law.  The remaining cases follow from the identity $(X\times_Y X)^{\times_X^i}\simeq X^{\times^{i+1}_Y}$ (which is itself a consequence of the pasting law) and induction. 
   \begin{equation}
       \begin{tikzcd}
          &X\times_Y X\ar[d]\ar[r]\ar[dr,phantom,very near start, "\lrcorner"]&Y \ar[d]\\
          X\ar[r]&X\times X\ar[r]&Y\times Y
       \end{tikzcd}
   \end{equation}
\end{proof}

\noindent Suppose that $f:X\to Y$ is a map in $\prstk$.  There is a correspondence 
\begin{equation}
    \mathcal{L}Y:=Y\times_{Y\times Y}Y\xleftarrow{\pi:=f\times_{f\times \id}f}X\times_{X\times Y}X\xrightarrow{\delta:=\id\times _{\pi_Y}\id}X\times_YX
\end{equation}
where $\pi_Y:X \times Y\rightarrow Y$ is a projection.  Note that there is also an identification $\mathcal{L}Y\times_Y X\simeq X\times _{X\times Y}X$. The following is a straightforward generalization of \cite[Theorem 5.3]{ben2010integral}, and we will refer there for further details.  
\begin{theorem}\label{Thm:BFNCenter}
   Suppose that $f:X\to Y$ is a map in $\prstk$ satisfying $\nQcoh$-descent, and $Y$ is $n$-affine.  There is an equivalence of braided categories
   \begin{equation}
       \FZ(\nQcoh^*(X\times_Y X))\simeq \nQcoh^\star(\mathcal{L} Y). 
   \end{equation}
   The canonical map $\FZ(\nQcoh^\star(X\times_YX))\to \nQcoh^\star(X\times_YX)$ is identified with 
   \begin{equation}
\delta_*\pi^*:\nQcoh^\star(\mathcal{L}Y)\to \nQcoh^\star(X\times_YX). 
   \end{equation}
\end{theorem}
\begin{proof}
    Denote $A=\nQcoh^\star(X\times_YX)$ and $B=\nQcoh^\otimes(X)$.  We first introduce two identifications necessary for our proof.  We consider the Bar construction of $A$: There is a simplicial complex $C^\bullet$ with $C^i\simeq A^{\times_B^{i+2}}$ the $(i+2)$-fold relative product over $B$, and $\colim_{\Delta^{\op}}C^\bullet\simeq A$.  Note that there is an induction/restriction adjunction 
    \begin{equation}\label{prf:ctr1}
        \begin{tikzcd}
            \Bimod_{B|A}\ar[r,shift left=2,"A\otimes_B-"{name=A}]&\Bimod_{A|A}\ar[l,shift left=2,""{name=B}]\ar[from=A,to=B,phantom,"\dashv"rotate=-90]
        \end{tikzcd}
    \end{equation}
Using (respectively) \cite[Theorem 5.3.1.30]{LurieHA}, the Bar construction, and the adjunction in \Cref{prf:ctr1},
    we have the equivalence:
    \begin{equation}
    \FZ(A)\simeq\Hom_{\Bimod_{A|A}}(A,A)\simeq \lim \Hom_{\Bimod_{A|A}}(A^{\times^{i+2}_B},A)
        \simeq \lim \Hom_{\Bimod_{B|A}}(A^{\times_B^{i+1}},B)\,.
    \end{equation}
    Furthermore, applying
     \Cref{Prop:BFN} and \Cref{lemma:BFNPullback}, we obtain the equivalence
    \begin{align}
         \lim \Hom_{\Bimod_{B|A}}(A^{\times_B^{i+1}},B)
        &\simeq \lim\nQcoh^\star(X\times_{X\times X}(X\times_YX)^{\times^{i+1}_X})\\ \notag
        &\simeq\lim \nQcoh^\star(\mathcal{L}Y\times_YX^{\times_Y^{i+1}})\\
        &\simeq \nQcoh^\star(\mathcal{L}Y), 
    \end{align}
    where the final step follows since the right hand side is identified with $\nQcoh$ of the \v{C}ech complex of $\mathcal{L}Y\times_YX\to \mathcal{L}Y$.  Since $X\to Y$ satisfies descent, so too does the map $\mathcal{L}Y\times_Y X\to \mathcal{L}Y$. 

Under the above identification, the canonical functor $\FZ(\nQcoh^\star(X\times_YX))\to \nQcoh^\star(X\times_YX)$ corresponds to  
%\dev{I think there is a typo in the BFN paper that needs to be sorted }
\begin{equation}
    \nQcoh^\star(\mathcal{L}Y)\xrightarrow{\pi^*}\nQcoh^\star(\mathcal{L}Y\times_YX)\simeq \nQcoh^\star(X\times_{X\times Y} X)\xrightarrow{\delta_*}\nQcoh^\star(X\times_YX). 
\end{equation}
\end{proof}

%%%%%%%%%%%%%%%%%%%%%%%%%%%%%%%%%%%%%%%%
\subsection{Morita Equivalences}\label{subsection:morita}
%%%%%%%%%%%%%%%%%%%%%%%%%%%%%%%%%%%%%%%%
We now explain the Morita equivalence between two types of modules for higher quasi-coherent sheaves. The first type uses the symmetric monoidal structure on $\nQcoh$, and the latter uses the convolution monoidal structure. The physical contents as it pertains to gauge theory was discussed in \S\ref{subsection:extractingelectric}.  The following generalizes \cite[Theorem 1.3]{BFNMorita}. 
\begin{theorem}\label{thm:MoritaEquivalence}
    Suppose that $X,Y$ are prestacks, $Y$ is $n$-affine, and $f:X\to Y$ satisfies $\nQcoh$-descent.  Suppose that $n\geq 2$ \underline{\textbf{or}} $n=1$ and $f: X\to Y$ is affine-schematic.  There is a Morita equivalence 
    \begin{equation}
        \varphi:\Mod_{\nQcoh^\otimes(X)}\xrightarrow{\sim} \Mod_{\nQcoh^\star(X\times_YX)}. 
    \end{equation}
\end{theorem}

\begin{proof}
    We aim to show that the functor $\varphi:=\Fun_{\nQcoh^\otimes(X)}(\nQcoh^\otimes(Y),-):\Mod_{\nQcoh^\otimes(X)}\to \mathbf{(n+1)}\Pr^{\mathrm{L}}_\st$ is monadic.  Tensoring by $\nQcoh^\otimes(X)$ provides a left adjoint.  By \Cref{cor:nQCohDual}, we see that $\nQcoh^\otimes(X)$ is self dual over $\nQcoh^\otimes(Y)$.  Then $\varphi\simeq \nQcoh^\otimes(Y)\otimes _{\nQcoh^\otimes(X)}(-)$ preserves limits and colimits.  We wish to show that $\varphi$ is conservative.  It suffices to show that $\nQcoh(Y)$ can be written as a colimit of objects in the image of the adjoint to $\varphi$.  

By enriched Barr-Beck-Lurie \cite[Theorem 7.4.10]{stefanich:thesis}, we conclude that $\varphi$ is monadic,  and so it factors as an isomorphism through $\Mod_{\End^\mathrm{L}_{\nQcoh^\otimes(Y)}(\nQcoh^\otimes(X))}$.  The result follows by using Proposition \ref{Prop:BFN} to identify $\End^{\mathrm{L}}_{\nQcoh^\otimes(Y)}(\nQcoh^\otimes(X))\simeq \nQcoh^\star(X\times_Y X)$.  
   
\end{proof}

%%%%%%%%%%%%%%%%%%%%%%%%%%%%%%%%%%
\section{Conclusion and Discussion}\label{section:discussion}
%%%%%%%%%%%%%%%%%%%%%%%%%%%%%%%%%%
We have highlighted the utility of geometric categories by developing a framework which recreates the symmetries of $G$-gauge theories.  This is accomplished by first equivariantizing the $G$-action to obtain the category $\nRep(G)$, from which we build a TQFT in $(n+2)$d.  $\nRep(G)$-modules correspond to boundaries of the TQFT, and their endomorphisms are then computed to determine the symmetry on the $(n+1)$d boundary.  We are able to identify the electric and magnetic symmetries of gauge theories in this context. 

This constitutes a nontrivial preliminary step toward a complete picture of categorical gauging.  Much remains to be understood, both on the mathematical and physical side.  We would like to recreate many of the results here for versions of $\nQcoh^\star(G)$, with multiplication twisted by $\omega\in \H^3(\rB G,\mathbb{G}_m)$, and explore Tambara-Yamagami in this language.

We have only discussed flat connections, avoiding the curvature which induces dynamics in the path-integral.  We would like to understand what additional data is associated to the physical boundary that accounts for dynamical gauging.  In light of \cite{StefanichSheaf}, many of the arguments of this work can be applied to any sheaf theory.  It would be desirable to construct a sheaf theory which encodes non-flat connections, and recreate the results there.  This would likely require a deeper understanding of internal higher algebra for the purpose of working within the smooth topos.  Perhaps in this case, we will observe that the operator that implements magnetic symmetry is directly related to the curvature of the gauge field.

%%%%%%%%%%%%%%%%%%%%%%%%%%%%%%%%
\section*{Acknowledgments}
%%%%%%%%%%%%%%%%%%%%%%%%%%%%%%%%%%%%%%
It is a pleasure to thank 
Andrea Antinucci,
André Henriques,
Patrick Kinnear,
Liang Kong,
 and Jackson Van Dyke
for helpful discussions. 
We thank Andrea Antinucci, Thomas Bartsch, Yuhan Gai, and Sakura Schäfer-Nameki for detailed feedback on an earlier draft.
We would like to thank the Isaac Newton Institute for Mathematical Sciences, Cambridge, for support and hospitality during the program:  Quantum field theory with boundaries, impurities, and defects, where work on this paper was undertaken.  DS would like to thank the University of Oxford for hosting him during the writing process.  This work was supported by EPSRC grant EP/Z000580/1.
DS is supported by VILLUM FONDEN, VILLUM Young Investigator grant 42125. 
MY is supported by the EPSRC Open Fellowship EP/X01276X/1.

\bibliographystyle{alpha}
\bibliography{ref.bib}
\end{document}